\documentclass[11pt]{article}
\usepackage{fullpage}

\usepackage{microtype}
\usepackage{graphicx}
\usepackage{subfigure}
\usepackage{booktabs} 

\usepackage{amsthm}
\usepackage{amssymb}
\usepackage{amsmath}
\usepackage{leftidx}
\usepackage[all]{xy}
\usepackage{microtype}
\usepackage[hidelinks]{hyperref}
\usepackage{tikz}
\usetikzlibrary{arrows}
\usepackage{listings}
\usepackage{graphicx}
\usepackage{todonotes}
\usepackage{float}
\newcommand\range{\mathrm{range}}
\newcommand\poly{\mathrm{poly}}
\newcommand\ranges{\mathrm{ranges}}
\newcommand\rng[1]{\mathrm{range}_{#1}}

\newcommand\nnz{\mathtt{nnz}}

\newcommand{\hide}[1]{}

\newtheorem{mydef}{Definition}[section]
\newtheorem{lem}[mydef]{Lemma}      
      
\newtheorem{thm}{Theorem}
\newtheorem{pro}[mydef]{Proposition}
\newtheorem{cor}[mydef]{Corollary}

\newtheorem{ass}{Assumption}[section]

\allowdisplaybreaks
\title{\vspace{-2em}Oblivious sketching for logistic regression}

\author{Alexander Munteanu\thanks{Dortmund Data Science Center, Faculties of Statistics and Computer Science, TU Dortmund University, Dortmund, Germany. Email: \texttt{alexander.munteanu@tu-dortmund.de}.}
\and Simon Omlor \thanks{Faculty of Statistics, TU Dortmund University, Dortmund, Germany. Email: \texttt{simon.omlor@tu-dortmund.de}.}
\and David Woodruff \thanks{Department of Computer Science, Carnegie Mellon University, Pittsburgh, PA, USA. Email: \texttt{dwoodruf@cs.cmu.edu}.}}

\begin{document}
\maketitle




\begin{abstract}
What guarantees are possible for solving logistic regression in one pass over a data stream? To answer this question, we present the first data oblivious sketch for logistic regression. Our sketch can be computed in input sparsity time over a turnstile data stream and reduces the size of a $d$-dimensional data set from $n$ to only $\operatorname{poly}(\mu d\log n)$ weighted points, where $\mu$ is a useful parameter which captures the complexity of compressing the data. Solving (weighted) logistic regression on the sketch gives an $O(\log n)$-approximation to the original problem on the full data set. We also show how to obtain an $O(1)$-approximation with slight modifications. Our sketches are fast, simple, easy to implement, and our experiments demonstrate their practicality.
\end{abstract}

\section{Introduction}
Sketches and coresets are arguably the most promising and widely used methods to facilitate the analysis of massive data with provable accuracy guarantees \cite{Phillips17,MunteanuS18,Feldman20}. 
Sketching has become a standard tool in core research areas such as data streams \cite{Muthukrishnan05} and numerical linear algebra \cite{Mahoney11,Woodruff14}, and constantly paves its way into diverse areas including computational geometry \cite{BravermanCKWY19,MeintrupMR19}, computational statistics \cite{GeppertIMQS17,Munteanu19}, machine learning \cite{Nelson20} and artificial intelligence \cite{BrandPSW20,GajjarM20,MolinaMK18}. Following the sketch-and-solve paradigm, we first apply a simple and fast dimensionality reduction technique to compress the data to a significantly smaller \emph{sketch} of polylogarithmic size. In a second step we feed the sketch to a standard solver for the problem, that needs little or no modifications. The theoretically challenging part is to prove an approximation guarantee for the solution obtained from the sketch with respect to the original large data set. 

\subsection{Related work}
\textbf{Deficiencies of coreset constructions.}
Most works on logistic regression have studied \emph{coresets} as a data reduction method. Those are small subsets of the data, often obtained by subsampling from a properly designed importance sampling distribution \cite{HugginsCB16,TolochinskiF18,MunteanuSSW18,TukanMF20,SamadianPMIC20}.
Those results often rely on regularization as a means to obtain small coresets. This changes the sampling distribution such that they do not generally apply to the unregularized setting that we study. The above coreset constructions usually require random access to the data and are thus not directly suitable for streaming computations. Even where row-order processing is permissible, at least two passes are required, one for calculating or approximating the probabilities and another for subsampling and collecting the data, since the importance sampling distributions usually depend on the data. A widely cited general scheme for making static (or multi-pass) constructions streamable in one pass is the merge \& reduce framework \cite{BentleyS80}. However, this comes at the cost of additional polylogarithmic overhead in the space requirements and also in the update time. The latter is a severe limitation when it comes to high velocity streams that occur for instance in large scale physical experiments such as the large hadron collider, where up to 100\,GB/s need to be processed and data rates are anticipated to grow quickly to several TB/s in the near future \cite{rohr2018data}. While the amortized insertion time of merge \& reduce is constant for some problems, in the worst case $\Theta(\log n)$ repeated coreset constructions are necessary for the standard construction to propagate through the tree structure; see e.g. \cite{FeldmanSS20}. This poses a prohibitive bottleneck in high velocity applications. Any data that passes and cannot be processed in real time will be lost forever.

Another limitation of coresets and the merge \& reduce scheme is that they work only in insertion streams, where the data is presented row-by-row. However it is unclear how to construct coresets when the data comes in column-wise order, e.g., when we first obtain the incomes of all individuals, then receive their heights and weights, etc. A similar setting arises when the data is distributed \emph{vertically} on numerous sites \cite{StolpeBDM13}. Sensor networks are another example where each sensor is recording only a single or a small subset of features (columns), e.g., each at one of many different production stages in a factory. Also the usual form of storing data in a table either row- or column-wise is not appropriate or efficient for extremely massive databases. The data is rather stored as a sequence of $(key, value)$ pairs in an arbitrary order in big unstructured databases \cite{GessertWFR17,SiddiqaKG17}.

The only work that can be simulated in a turnstile stream to tackle the extreme settings described above, is arguably \cite{SamadianPMIC20} via uniform subsampling. Their coreset size is roughly $\Theta(d\sqrt n)$ and works only when the problem is regularized very strongly such that the loss function is within constant factors to the regularizer, and thus widely independent of the input data. Consequently, the contribution of each point becomes roughly equal and thus makes uniform sampling work. However, those arguments do not work for unconstrained logistic regression, where each single point can dominate the cost and thus no sublinear compression below $\Omega(n)$ is possible in the worst case, as was shown in \cite{MunteanuSSW18}. To cope with this situation, the authors of \cite{MunteanuSSW18} introduced a complexity parameter $\mu(A)$ that is related to the statistical modeling of logistic regression, and is a useful measure for capturing the complexity of compressing the dataset $A$ for logistic regression. They developed a coreset construction of size $\tilde O(\mu d^{3/2}\sqrt{n})$. Although calculating their sampling distribution can be simulated in a row-order stream, the aforementioned limitation to two passes is an unsolved open problem. The coreset size was reduced to $\poly(\mu d \log n)$ but only at the cost of even more row-order passes to compute repeatedly a coreset from a coreset, $O(\log\log n)$ times.\\

\textbf{On the importance of data oblivious sketching.}
Oblivious sketching methods are much better positioned for handling high velocity streams, as well as highly unstructured and arbitrarily distributed data. Linear sketches allow efficient applications in single pass sequential streaming and in distributed environments, see, e.g. \cite{ClarksonW09,WoodruffZ13,KannanVW14}. Linear sketches can be updated in the most flexible dynamic setting, which is commonly referred to as the \emph{turnstile} model, see, e.g., \cite{Muthukrishnan05} for a survey. In this model we initialize a matrix $A$ to the all-zero matrix. The stream consists of $(key, value)$ updates of the form $(i, j, v)$, meaning that $A_{ij}$ will be updated to $A_{ij} + v$. A single entry can be defined by a single update or by a subsequence of not necessarily consecutive updates. For instance, a sequence $\ldots,(i,j,27),\ldots,(i,j,-5),\ldots$ will result in $A_{ij} = 22$. Deletions are possible in this setting by using negative updates matching previous insertions. At first glance this model might seem technical or unnatural but we stress that for dealing with the aforementioned unstructured data, the design of algorithms working in the turnstile model is of high importance. We will see how any update can be calculated in $O(1)$ basic operations so it becomes applicable in high velocity real-time applications. Additionally, due to linearity, oblivious sketching algorithms can be represented as linear maps, i.e., sketching matrices $S$. In particular they support several operations such as adding, subtracting, and scaling databases $A_j$ efficiently in the sketch space, since $SA=S\sum_{j}\alpha_j A_j = \sum_{j}\alpha_j SA_j$. For instance, if $A_{t_1}$ and $A_{t_2}$ are balances of bank accounts at time steps $t_1 < t_2$, then $SB=SA_{t_2}-SA_{t_1}$ is a sketch of the changes in the period $t\in (t_1, t_2]$.\\

\textbf{Data oblivious sketching for logistic regression.}
In this paper we deal with unconstrained logistic regression in one pass over a turnstile data stream. As most known turnstile data stream algorithms are linear sketches (and there is some evidence that linear sketches are optimal for such algorithms in certain conditions \cite{LiNW14,AiHLW16}), it is natural for achieving our goals to look for a distribution over random matrices that can be used to sketch the data matrix such that the (optimal) cost of logistic regression is preserved up to constant factors. Due to the aforementioned impossibility result, the reduced sketching dimension will depend polynomially on the mildness parameter $\mu(A)$, and thus we need $\mu(A)$ to be small, which is common under usual modeling assumptions in statistics \cite{MunteanuSSW18}. In this setting, logistic regression becomes similar to an $\ell_1$-norm regression problem for the subset of misclassified inputs, and a uniform sample suffices to approximate the contribution of the other points. 

Known linear \emph{subspace embedding} techniques for $\ell_1$ based on Cauchy ($1$-stable) random variables \cite{SohlerW11} or exponential random variables \cite{WoodruffZ13} have a dilation of $O(d\log d)$ or higher polynomials thereof, and nearly tight lower bounds for this distortion exist \cite{WangW19}. While a contraction factor of $(1-\varepsilon)$ seems possible over an entire linear subspace, a constant dilation bound for an arbitrary but fixed vector (e.g., the optimal solution) are the best we can hope for \cite{Indyk06,ClarksonW15,lwy21}. A general sketching technique was introduced by \cite{ClarksonW15} that achieves such a \emph{lopsided} result for all regression loss functions that grow at least linearly and at most quadratically (the quadratic upper bound condition is necessary for a sketch with a sketching dimension sub-polynomial in $n$ to exist \cite{BravermanO10a}) and have properties of norms such as symmetry, are non-decreasing in the absolute value of their argument, and have $f(0)=0$, which is true for a class of robust $M$-estimators, though not for all. For example, the Tukey regression loss has zero growth beyond some threshold. The above sketching technique has been generalized to cope with this problem \cite{ClarksonWW19}. However the latter work still relies on a symmetric and non-decreasing loss function $f$ with $f(0)=0$.

For the logistic regression loss $\ell(v)=\ln(1+\exp(v))$, we note that it does not satisfy the above norm-like conditions since $\ell(0)=\ln(2)$, $\ell(x)\neq \ell(-x)$, and while it is linearly increasing on the positive domain, it is decreasing exponentially to zero on the negative domain. Indeed, the class of monotonic functions has linear lower bounds for the size of any coreset and more generally for any sketch \cite{TolochinskiF18,MunteanuSSW18}, where the unboundedness of the ratio $\ell(x)/ \ell(-x)$ plays a crucial role.

\subsection{Our contributions}
In this paper we develop the first oblivious sketching techniques for a generalized linear model, specifically for logistic regression. Our \texttt{LogReg}-sketch is algorithmically similar to the $M$-estimator sketches of \cite{ClarksonW15,ClarksonWW19}. However, there are several necessary changes and the theoretical analyses need non-trivial adaptations to address the special necessities of the logistic loss function. The sketching approach is based on a combination of subsampling at different levels and hashing the coordinates assigned to the same level uniformly into a small number of buckets \cite{IndykW05,VerbinZ12}. Collisions are handled by summing all entries that are mapped to the same bucket, which corresponds to a variant of the so-called \texttt{CountMin}-sketch \cite{CormodeM05}, where the sketch $S_h$ on each level is presented only a fraction of all coordinates.

More precisely, we define an integer branching parameter $b$ and a parameter $h_{\max}=O(\log_b n)$, and each row of our data matrix gets assigned to level $h \leq h_{\max}$ with probability proportional to $b^{-h}$.
The row is then assigned one of the $N$ buckets on level $h$ uniformly at random and added to that bucket. The new matrix that we obtain consists of $h_{\max}$ blocks, where each block consists of $N$ rows. The weight of a row is proportional to $b^h$.
The formal definition of the sketch is in Section \ref{secapprox}. This scheme is complemented by a row-sampling matrix $T$ which takes a small uniform sample of the data, which will be dealt with in Section \ref{sectUniform}.
\[ S = \left[ \begin{matrix} S_0 \\ S_1 \\ \vdots \\ S_{h_{\max}} \\ T \end{matrix} \right] \]
The intuition behind this approach is that coordinates are grouped according to \emph{weight classes} of similar loss which can be handled separately in the analysis. Weight classes with a small number of members will be approximated well on sketching levels with a large number of elements since roughly all members need to be subsampled to obtain a good estimate. Weight classes with many members will be approximated well on levels with a smaller number of subsamples, because if too many members survive the subsampling there will also be too many collisions under the uniform hashing, which would either lead to a large overestimate when those add up, or, due to asymmetry, would cancel each other and lead to large underestimations. The asymmetry problem is also one of the main reasons why we need a new analysis relying on the \texttt{CountMin}-sketch as a replacement for the \texttt{Count}-sketch previously used in \cite{ClarksonW15,ClarksonWW19}. The reason is that \texttt{Count}-sketch assigns random signs to the coordinates before summing them up in a bucket. The error could thus not be bounded if the sign of an element is changed since the ratio $\ell(x)/ \ell(-x)$ is unbounded for unconstrained logistic regression. Finally, since there could be too many small contributions near zero and logistic regression, unlike a normed loss function, assigns a non-zero but constant loss to them, their contribution can become significant. This is taken care of by the small uniform sample of size $\tilde O(\mu d)$.

Our main result is the following theorem, where $\nnz(A)$ denotes the number of non-zero entries in $A$ or in a data stream it corresponds to the number of updates, 
\[  f_w(A x)=\sum\limits_{i\in [n]} w_i\cdot \ln\left(1+\exp({a_i x})\right) \] denotes the weighted logistic loss function, and $f(Ax)$ is the unweighted case where $w$ is the all $1$s vector. It also assumes that the data is $\mu$-complex for a small value $\mu$ meaning that $\mu(A)\leq \mu$ as in \cite{MunteanuSSW18}, see Section \ref{sec:prelim} for a formal definition:
\begin{thm}\label{mainthm2}
Let $A \in \mathbb{R}^{n \times d}$ be a $\mu$-complex matrix for bounded $\mu< \infty $.
Then there is a distribution over sketching matrices $S\in \mathbb{R}^{r \times n}$ with
$r = \poly\left(\mu d \log(n)\right)$, and a corresponding weight vector $w \in \mathbb{R}^r$, for which $B=SA$ can be computed in $O(\nnz(A))$ time over a turnstile data stream and for which if $x'$ is the minimizer to $\min_x f_w(Bx)$, then with constant probability it holds that 
\begin{align*}
f(Ax')\leq {O}(\log n) \min_{x\in \mathbb{R}^d} f(Ax).
\end{align*}
Further, there is a convex function $f_{w, c}$ such that for the minimizer $x''$ to $\min_x f_{w, c}(Bx)$ it holds that 
\begin{align*}
f(Ax'')\leq {O}(1) \min_{x\in \mathbb{R}^d} f(Ax)
\end{align*}
with constant probability.
\end{thm}

The first item is a sketch-and-solve result in the sense that first, the data is sketched and then the sketch is put into a standard solver for weighted logistic regression. The output is guaranteed to be an $O(\log n)$ approximation. The second item requires a stronger modification which can be handled easily for instance with a subgradient based solver for convex functions. The individual loss for $f_{w,c}$ remains the original logistic loss as in $f_w$ for each point. However for a fixed $x$ occurring in the optimization, the loss and gradient are evaluated only on the $K$ largest entries on each level of the sketch $Bx$ (except for the uniform sample $T$), for a suitable $K<N$. This preserves the convexity of the problem, and guarantees a constant approximation. The details are given in Sections \ref{sec:dilation}, and \ref{sec:logreg}.\\

\textbf{Overview of the analysis.}
The rest of the paper is dedicated to proving Theorem \ref{mainthm2}. We first show that the logistic loss function can be split into two parts $f(Ax)\approx G^+(Ax)+f((Ax)^-)$, which can be handled separately while losing only an approximation factor of two; see Section \ref{sec:prelim}.

The first part is $G^+(y):=\sum_{y_i \geq 0}y_i$, the sum of all positive entries which can be approximated by the aforementioned collection of sketches $S_h$.
Here we show that with high probability no solution becomes much cheaper with respect to the objective function and that with constant probability, the cost of some good solution does not become too much larger, which can be bounded by a factor of at most $O(\log n)$ or $O(1)$ depending on which of our two algorithms we use.
We prove this in Section \ref{secapprox}. First, we bound the contraction. To do so we define \emph{weight-classes} of similar loss. For weight classes with a small number of members, a leverage-score argument yields that there cannot be too many influential entries. For larger weight classes, we show that there exists a small subset of influential entries that on some subsampling level do not collide with any other influential entry when hashing into buckets, and thus represent their weight class well. This concludes the handling of the so-called \emph{heavy-hitters}, see Section \ref{sec:contraction}. Those arguments hold with very high probability, and so we can union bound over a net and relate the remaining points to their closest point in the net. This yields the contraction bound for the entire solution space, see Section \ref{sec:netargument}. Although the high-level outline is similar to \cite{ClarksonW15}, several non-trivial adaptations are necessary to deal with the assymmetric $G^+$ function that is not a norm and has zero growth on the negative domain.
The $O(\log n)$ dilation bound follows by a calculation of the expected value on each level and summing over $h_{\max}=O(\log n)$ levels. The $O(1)$ bound requires the aforementioned clipping of small contributions on each level. Each weight class $q$ makes a main contribution on some level $h(q)$. The argument is now that it can have a significant impact only on levels in a small region of size $k=O(1)$ around $h(q)\pm k$. Further,  with high probability for $h>h(q)+ k$ there will be no element of the same weight class, so that the contribution to the expectation is zero, and for $h<h(q)- k$ the contribution can be bounded by $O(h_{\max}^{-1})$, so that for all three cases the expected contribution is at most $O(1)$ after summing over all levels.

The second part is $f^-:=f((Ax)^-)$ which maps any misclassified point to $\ell(0)=\ln(2)$ and the remaining points to the usual logistic loss of a point, i.e., $\ell(a_ix)=\log(1+\exp(a_ix))$.
Here we prove that for $\mu$-complex data sets the worst case contribution of any point can be bounded by roughly $\mu/n$ and thus a uniform sample of $\tilde O(\mu)$ can be used to approximate $f^-$ well. This will be done via the well-known sensitivity framework \cite{LangbergS10} in Section \ref{sectUniform}. 
We put everything together to prove Theorem \ref{mainthm2} in Section \ref{sec:logreg}.
In Section \ref{sec:experiments} our experimental results demonstrate that our sketching techniques are useful and competitive to uniform sampling, SGD, and an adaptive coreset construction. We show in some settings the oblivious sketch performs almost the same or better, but is never much worse. We stress that neither SGD nor the coreset allow the desired turnstile streaming capabilities.
We finally conclude in Section \ref{sec:conclusion}.\\
{Omitted proofs and details can be found in the supplementary material.}


\section{Preliminaries}\label{sec:prelim}

\subsection{Notation}

In logistic regression we are usually given a data matrix $X\in\mathbb{R}^{n\times d}$ and a label vector $L\in \{-1,1\}^n$. For notational brevity and since a data point always appears together with its label, we technically work with a data matrix $A\in\mathbb{R}^{n\times d}$ where each row $a_i$ for $i\in [n] $ is defined as $a_i:=-l_ix_i$. We set $g(v)=\ln(1+\exp({v}))$ for $v \in \mathbb{R}$. Our goal is to find $x \in \mathbb{R}^d$ that minimizes the logistic loss given by
\[  f(A x)=\sum\limits_{i\in [n]} g({a_i x})=\sum\limits_{i\in [n]} \ln\left(1+\exp({a_i x})\right).  \] 

We parameterize our analysis by
\[\mu_A=\sup_{x \in \mathbb{R}^d\setminus\{0\}}\frac{\Vert (A x)^+ \Vert_1}{\Vert( A x)^- \Vert_1}
\]
where for $y \in \mathbb{R}^n$, the vector $y^+$ (resp. $y^-$) denotes the vector with all negative (resp. positive) entries replaced by $0$. This definition of $\mu$ has been introduced before in \cite{MunteanuSSW18} and is a useful parameter to bound the amount of data reduction possible for logistic regression.
In the remainder we assume that $A$ is $\mu$-complex, meaning that $\mu_A\leq \mu$ for some $1\leq \mu \leq \infty$. 
For any vector $y$ we further define $G^+(y):=\sum_{y_i \geq 0}y_i$ to be the sum of all positive entries. Also we define $G(y):=\|y\|_1$. 
Note that by definition of $\mu$ the supremum considers for each $x\in \mathbb{R}^d$, also $-x$. Therefore, it holds for all $x \in \mathbb{R}^d$ that
\[\mu^{-1} \Vert (A x)^- \Vert_1 \leq \Vert (A x)^+ \Vert_1 \leq \mu \Vert (A x)^- \Vert_1.\]
In particular, the property 
$G^+(Ax)=\Vert (A x)^+ \Vert_1 \geq \frac{\Vert (A x)^- \Vert_1}{\mu}$ will often be used.

\subsection{Initial approach}

Our first idea is that we can split $f$ into two parts which can be approximated independently.

\begin{lem}\label{splitlemma}
For  all $x \in \mathbb{R}^d$ it holds that 
\[ f(A x)~\geq~ \frac{1}{2}\left(f\left((A x)^-\right)+G^+(A x)\right) \]
and
\[
f(A x) ~\leq~  f\left((A x)^-\right)+G^+(A x) . \]
\end{lem}

{
\begin{proof}
Let $v \in \mathbb{R}_{\geq 0}$. Then it holds that $g(v)=g(0)+\int_{0}^{v}g'(y)dy=\ln(2)+\int_{0}^{v}g'(y)dy$. Note that $g'(y)=\frac{\exp(y)}{1+\exp(y)}$, and thus for any $y \in \mathbb{R}$ we have $g'(y) \leq 1$ and for any $y \in [0, \infty )$ we have $g'(y)\geq \frac{\exp(y)}{2\exp(y)}=\frac{1}{2} $. We conclude that
\[ g(v)=g(0)+\int_{0}^{v}g'(y)dy \geq g(0)+\int_{0}^{v}\frac{1}{2} dy=g(0)+\frac{1}{2}v  \]
and
\[ g(v)=g(0)+\int_{0}^{v}g'(y)dy \leq g(0)+\int_{0}^{v}1 dy=g(0)+v . \]
Recall that in $(A x)^-$ each coordinate $a_ix>0$ is replaced by zero. Thus, if $a_ix>0$ then $g(a_ix)=g(0)+\int_{0}^{a_ix}g'(y)dy=g((Ax)^-_i)+\int_{0}^{a_ix}g'(y)dy$.
Hence, it holds that 
\begin{align*}
f(A x)&=\sum_{a_ix < 0}g(a_ix) ~+~ \sum_{a_ix \geq 0}g(a_ix) \\
&\leq \sum_{a_ix < 0}g(a_ix)  ~+~  \sum_{a_ix \geq 0}g(0)+ a_ix\\
&=\sum_{a_ix < 0}g(a_ix)  ~+~  \sum_{a_ix \geq 0}g(0) ~+~ \sum_{a_ix \geq 0}a_ix\\
&=  f((A x)^-)+G^+(A x) 
\end{align*}
and similarly
\begin{align*}
f(A x)&=f((A x)^-)+\sum_{a_i x>0}\int_{0}^{a_i x}g'(y)dy \\
&\geq  f((A x)^-)+\frac{1}{2} G^+(A x)\\
&\geq \frac{1}{2}\left(f\left((A x)^-\right)+G^+(A x)\right) . 
\end{align*}
\end{proof}
}

Next we show that for $\min_{x \in \mathbb{R}}f(A x)$, there is a non-trivial lower bound that will later be used to show that $ f((A x)^-)$ can be approximated well enough.

\begin{lem}\label{minvalue}
For all $x\in \mathbb{R}^d$ it holds that 
\[f(A x) \geq \frac{n}{2\mu} \left(1+\ln(\mu)\right) = \Omega \left(\frac{n}{\mu}(1+\ln(\mu))\right) .\]
\end{lem}
{
\begin{proof}
For any $w \geq 1$ it holds that $ \ln(w)=\int_{1}^{w}\frac{1}{y}dy$.
Thus for $v\leq 0$ we have $g(v)=\ln(1+e^{v})=\int_{1}^{1+e^{v}}\frac{1}{y}dy\geq \frac{ e^v}{2}$ since for $1 \leq y \leq e^{v}+1 \leq 2$ we have $\frac{1}{y} \in [\frac{1}{2}, 1]$.
Let $z=A x$.
Using this fact and Lemma \ref{splitlemma} we get $f(z) \geq \frac{1}{2}(\sum_{i} \exp(\min\{z_i, 0\})+G^+(z))$.
Since $\exp(v)$ is convex, Jensen's inequality implies $\sum_{i} \exp(\min\{z_i, 0\})=n\sum_{i} \frac{1}{n} \exp(\min\{z_i, 0\}) \geq  n\exp(\frac{1}{n}\sum_{i} \min\{z_i, 0\})$.
Using this argument we get for $y=\frac{\Vert z^-\Vert_1}{n}$ that $ \sum_{i} \exp(\min\{z_i, 0\})\geq n \exp(-y)$.
Recall that $G^+(z)\geq \frac{yn}{\mu}$ holds by definition of $\mu$.

Using Lemma \ref{splitlemma} we conclude that $f(z)\geq \frac{1}{2}(n\exp(-y)+\frac{yn}{\mu})$.
The function $(n\exp(-y)+\frac{yn}{\mu})$ is minimized over $y$ if its first derivative is zero, i.e., if
\[ n\exp(-y)=\frac{n}{\mu} \]
which is equivalent to $y=\ln(\mu )$.
Hence $f(z)\geq \frac{1}{2} \left(\frac{n}{\mu}+\frac{n \ln(\mu)}{\mu}\right)$. \qedhere
\end{proof}
}

\section{Approximating \texorpdfstring{$G^+(A x)$}{G+(A x)}}\label{secapprox}

Here we focus on approximating $G^+(A x)=\sum_{a_i x>0}a_i x$.
We develop a sketching method similar to the approach of \cite{ClarksonW15}.
This gives us a new matrix $A'=S A$, referred to as the sketch, for which we will show that with high probability it holds for all $x \in \mathbb{R}^d$, that $G^+(A' x)\geq (1-\varepsilon) G^+(A x)$ for $\varepsilon>0$, and we also have $\mathbb{E} (G^+(A' x))\leq C G^+(A x)$ for some constant $C>1$.
We show the following result:
\begin{thm}\label{ThmG+}
Given $\varepsilon>0$, and $\delta>0$ we set 
$r={O}\left( d^5 \left(\frac{\mu}{\varepsilon}\right)^{7}\delta^{-2}\ln^4\left(\frac{n \mu}{\delta \varepsilon}\right) \right)$.
Then there is a random matrix $S \in \mathbb{R}^{r \times n}$ such that for $A'=SA$ and a convex function $G^+_c$, it holds that $G^+_c(A' x)\geq (1-\varepsilon) G^+(A x)$ and $\mathbb{E} (G^+_c(A' x))\leq C G^+(A x)$ for some constant $C>1$ and for all $x \in \mathbb{R}^d$. The failure probability of this event is at most $\delta$.
\end{thm}

\subsection{The sketching algorithm}

The idea is to hash the rows of $A$ uniformly into buckets.
The rows that are assigned to the same bucket are added to obtain a row of $A'$ which can also be written as $A'=SA$ for a suitable matrix $S$.
To avoid that for a given $z=Ax$, there are too many entries of $z$ that cancel with each other, we assign a level to each bucket.
The level of a bucket determines how many coordinates are assigned to it (in expectation).
Buckets with fewer coordinates are given a larger weight.
In this way, large entries of $Ax$ are preserved in buckets with many coordinates, up to a small error, while the contribution of many small entries of $Ax$ is preserved by buckets with few coordinates, but high weights.

More precisely, the sketching algorithm is defined as follows: we define $N$ to be the number of buckets in each level and $h_{\max}$ to be the number of levels.
Let $b$ be a branching parameter that determines how the (expected) number of coordinates changes between different levels.

Then each coordinate $p \in [n]$ is hashed uniformly at random to bucket $g_p \in [N]$ at level $h_p \in [h_{\max}]$, where we set $h_p=h$ with probability $\frac{1}{\beta b^h}$, for $0 \leq h \leq h_{\max}=\log_b(\frac{n}{m})$ and some $b>2$ and $\beta=\frac{b-b^{-h_{\max}}}{b-1}$.
The weight of $z_p$ is given by $w_p=b^{h_p}\beta$.
The sketching matrix is given by $S \in \mathbb{R}^{h_{\max} N \times n}$, where $(S)_{jp}=b^{h_p}\beta$ if $j=g_p+h_p N$ and $(S)_{jp}=0$ otherwise.

Assume we are given some error parameter $\varepsilon' \in (0, \frac{1}{3})$ and set $\varepsilon=\frac{\varepsilon'}{\mu'}$, where $\mu'=\mu+1$.
Then we have $ G^+(z) \geq \frac{G(z^-)}{\mu}=\frac{G(z)-G(z^+)}{\mu}$, which is equivalent to $\frac{(\mu +1)G^+(z)}{\mu}=G^+(z)+\frac{G^+(z)}{\mu}\geq \frac{G(z)-G(z^+)}{\mu}  $.
Multiplying by $\mu$ gives us $ G^+(z) \geq \frac{G(z)}{\mu+1}= \frac{G(z)}{\mu'}$.
Let $\delta<1$ be a failure probability.
Let $m$ be a parameter which determines whether a set of coordinates is considered large.

\subsection{Outline of the analysis}

Instead of explaining how the sketch is applied to $A$, we will explain how the sketch is applied to $z:=Ax$ for a fixed $x$. Note that $(SA)x=S(Ax)$.
We assume without loss of generality that $G(z)= \Vert z \Vert_1=1 $.
This can be done since for any $\lambda>0$, we have $G(S \lambda z)=\lambda G(S  z)$ and $G^+(S \lambda z)=\lambda G^+(S  z)$.

We split the entries of $z$ into weight classes and derive bounds for the contribution of each individual weight class.
The goal is to show that for each $x \in \mathbb{R}^d$ the entries of $z$ that can be large are the same, and for the remaining entries we can find a set of representatives which are in buckets of appropriate weights and are large in contrast to the remaining entries in their buckets.
Therefore we let $Z=\{z_p ~|~ p\in [n]\}$ be the multiset of values appearing in $z$.
We define weight classes as follows:\\
For $q \in \mathbb{N}$ we set $W_q^+=\{ z_p \in Z ~|~ 2^{-q}< z_p \leq 2^{-q+1} \}$ to be the positive weight class of $q$.
Similarly we define $W_q=\{ z_p \in Z ~|~ 2^{-q}< |z_p| \leq 2^{-q+1} \}$ to be the weight class of $q$.
Since we are also interested in the number of elements in each weight class we define $h(q):=\lfloor \log_b(\frac{|W_q^+|}{\beta m}) \rfloor$ if $|W_q^+|\geq \beta m $ and $h(q)=0$ otherwise.
This way we have $\beta m b^{h(q)} \leq |W_q^+| \leq \beta m b^{h(q)+1} $ and thus $h(q)$ is the largest index such that the expected number of entries from $W_q$ at level $h$ is at least $\beta m$.
Note that the contribution of weight class $W_q$ is at most $2^{-q+1} n$.
Thus we set $q_{\max}=\log_2 (\frac{n}{\varepsilon})$ and will ignore weight classes with $q > q_{\max}$ as their contribution is smaller than $\varepsilon$.

Our first goal will be to show that there exists an event $\mathcal{E}$ with low failure probability (which will be defined later) such that if $\mathcal{E}$ holds then $G^+(SAx)$ gives us a good approximation to $G^+(Ax) $ with very high probability.
More precisely:

\begin{thm}\label{Thm3.1}
If $\mathcal{E}$ holds then we have $G^+( SAx )\geq (1-60\varepsilon')G^+(Ax)$ for any fixed $x \in \mathbb{R}^d$ with failure probability at most $e^{-m\varepsilon^2/2}$.
\end{thm}

This will suffice to proceed with a net argument, i.e., we show that there exists a finite set $N \subset \mathbb{R}^d$ such that if we have $G^+( SAx )\geq (1-\varepsilon')G^+(Ax)$ for all $x \in N$ then it holds that $G^+( SAx )\geq (1-4\varepsilon')G^+(Ax)$ for all $x \in \mathbb{R}^d$, and thus we obtain the desired contraction bound.

\begin{thm}\label{netthm}
We have $G^+( SAx )\geq (1-240\varepsilon')G^+(Ax)$ for every $x \in \mathbb{R}^d$ with failure probability at most $2\delta$.
\end{thm}

Finally we show that in expectation $G^+( SAx )$ is upper bounded by $h_{\max}G^+( Ax )={O}(\log(n))$.
Further we show that there is a convex function $G_c^+$ and an event $\mathcal{E}'$ with low failure probability such that we have $G_c^+( SAx ) \geq  (1-\varepsilon') G^+( Ax )$ with high probability, and in expectation $G_c^+( SAx )$ is upper bounded by $C G^+( Ax )$ for constant $C$.

\begin{thm}\label{Thm3.2}
There is a constant $C>1$ such that if $\mathcal{E}'$ holds then
$\mathbb{E}(G_c^+(Sz)) \leq C G^+(z).$ 
\end{thm}





\section{Approximating \texorpdfstring{$f((A x)^-) $}{f((A x)-) }}\label{sectUniform}

The following theorem shows that a uniform sample $R \subset \{a_i ~|~ i \in [n]\}$ gives us a good approximation to $f((A x)^-) $.
We also show that in expectation the contribution of $R$ is not too large and thus with constant probability contributes at most $Cf(Ax)$ for some constant factor $C$, by Markov's inequality. Using the sensitivity framework and bounding its relevant parameters (see Section \ref{app:uniform}) we get:

\begin{thm}\label{Thm2.1}
For a uniform sample $R \subset \{a_i ~|~ i \in [n]\}$ of size $k = {O}(\frac{\mu}{\varepsilon^{2}}(d\ln(\mu)+\ln(\frac{1}{\delta_1}))$ we have that with failure probability at most $\delta_1$ that
\[ \sum_{a_i \in R}\frac{n}{k}g(a_i x) \geq f((A x)^-)-3\varepsilon f(A x) \]
for all $x\in \mathbb{R}^d$.
Further $\mathbb{E}(\sum_{a_i \in R}\frac{n}{k}g(a_i x))=f(A x)$.
\end{thm}

\section{Approximation for logistic regression in one pass over a stream of data}\label{sec:logreg}


We set $\varepsilon=\frac{1}{8}$.
To prove Theorem \ref{mainthm2} we first show that we can tweak the sketch $A'$ from Theorem \ref{ThmG+} by adding weights and scaling $A'$.
This way we get a new sketch where the weighted logistic loss and the sum of positive entries are the same up to an additive constant of $\ln(2)$.
More precisely we set $A''=N h_{\max} A'= N h_{\max} SA$ with equal weight $\frac{1}{N h_{\max}}$ for all rows.
We denote the weight vector by $w'$.
This way we make sure that $f_{w'}(A''x)$ is very close to $G^+(Ax)$:

\begin{lem}\label{lem5.2}
For all $x \in \mathbb{R}^d$ it holds that $G^+(A'x) \leq \frac{1}{N h_{\max}}f(A''x) \leq G^+(A'x)+\ln(2)$.
\end{lem}
{
\begin{proof}
First note that for any $v \geq 0$ we have $g(v)=\ln(1+\exp(v))\geq \ln(\exp(v))=v$.
This implies that 
\begin{align*}
f(A''x) & =\sum_{i=1}^{N h_{\max}} g(a_i''x) \geq \sum_{a_i''x \geq 0} g(a_i''x)
\geq \sum_{a_i''x \geq 0} a_i''x = \sum_{a_i'x \geq 0}N h_{\max} a_i'x = N h_{\max}G^+(A'x)
\end{align*}
Further, note that for any $v \in \mathbb{R}$ we have $g(v)\leq \ln(2)+\max\{v, 0\}$. 
This is true for $v\leq 0$ since $g$ is monotonically increasing and $g(0)=\ln(2)$, and since the derivative of $g$ is always bounded by $1$ it also holds for $v>1$, cf. Lemma \ref{splitlemma}.
Consequently it holds that 
\begin{align*}
f(A''x)  =\sum_{i=1}^{N h_{\max}} g(a_i''x)
\leq \sum_{i=1}^{N h_{\max}} \ln(2)+\max\{0, a_i''x\}
&= \sum_{i=1}^{N h_{\max}} \ln(2) ~+~\sum_{a_i''x \geq 0} a_i''x\\
&= {N h_{\max}} \ln(2)+\sum_{a_i'x \geq 0}N h_{\max} a_i'x\\
&= N h_{\max}(\ln(2)+G^+(Ax)).
\end{align*}
\end{proof}
}
We are now ready to show our main result:
\begin{proof}[Proof of Theorem \ref{mainthm2}]
Let $(T, u)$ be a weighted uniform random sample from Theorem \ref{Thm2.1}.
We define $B= \genfrac(){0pt}{0}{A''}{T}$ with weight vector $w=(w',u)$.
The size of $B$ is bounded by $\poly(\mu d \log n)$ since it is dominated by the sketch of Theorem \ref{ThmG+}.
Note that $B$ can handle any update in $O(1)$ time since we need to draw one random number for determining the level that a data point is assigned, another one for determining its bucket and a third one to decide whether to include it into the random sample or not. This sums up to $O(\nnz(A))$ time in total. We note that $B$ can be computed over a turnstile data stream when we replace the random number generators by hash maps, so to compute the pseudorandom choices on demand, see \cite{AlonBI86,Dietzfelbinger96}. 

Let $x'$ be a minimizer to $\min_x f_w(Bx)$.
By Theorem \ref{ThmG+}, Theorem \ref{Thm2.1}, Lemma \ref{lem5.2} and Lemma \ref{splitlemma} we have
\begin{align*}
~ f_w(Bx) &\geq f_u(Tx)+G^+(A'x)\\
&\geq f((Ax)^-) -3\varepsilon f(Ax) +(1-\varepsilon) G^+(Ax) \\
&\geq (1-4\varepsilon) f(Ax)\\
&=\frac{1}{2}f(Ax)
\end{align*}
with constant probability for all $x \in \mathbb{R}^d$ simultaneously.
Let $x^*$ be a solution minimizing $f(Ax)$.
By Theorem \ref{Thm2.1} and Lemma \ref{Lem3.12} it holds that 
\begin{align*}
&~\mathbb{E}\left(f_w(Tx^*)+G^+(A'x^*)\right) \leq \alpha\left(f\left((A'x^*)^-\right)+G^+(Ax^*)\right)
\end{align*}
 for $\alpha=O(\log(n))$
Thus using Markov's inequality we have 
\[f_w(Tx^*)+G^+(A'x^*)\leq 2\alpha\left(f((Ax^*))+G^+(Ax^*)\right)\]
with probability at least $1/2$.
Since for any $ \mu < \infty $ the optimal value of $f(Ax)$ is greater or equal to $\ln(2)$, i.e., there is at least one missclassification, the above inequality and Lemma \ref{splitlemma} imply
\begin{align*}
f(Ax') &\leq  2f_w(Bx')\\
&\leq  2\left( f_u(Tx')+G^+(A'x')+\ln(2)\right)\\
&\leq 2\left( f_u(Tx^*)+G^+(A'x^*)\right)+2\ln(2)\\
&\leq 2\left( 2\alpha\left(f(Ax^*)+G^+(Ax^*)\right)\right)+2\ln(2)\\
&\leq 10 \alpha f(Ax^*)
\end{align*}
with probability at least $1/2$, which proves the first part of the theorem.


%
%
%
%

The last part of the theorem can be derived by slightly changing the logistic loss function similar to a Ky Fan norm. More precisely $A''$ can be split into blocks $A_{h}$ for levels $h=0, \dots , h_{\max}$.
Then we define $f_{w,c}(Bx)=f_u(Tx)+f_{w', c}(A''x)$ where 
\[f_{w', c}(A''x):=\sum_h \sum_{i \in [K]}g((a''x)_{\pi(i, h)}) \]
where $(a''x)_{\pi(i, h)}$ denotes the $i$'th largest entry of $A_h x$.
The modified function thus omits for any fixed $x$ all but the largest $K$ entries on each level. Lemma \ref{lem5.2} can now be adjusted to show that $G_c^+(A'x) \leq f_{w', c}(A''x) \leq G_c^+(A'x)+\ln(2)$.
Then we can use Theorem \ref{Thm3.2} to show that $f(Ax'') \leq 10 C(Ax^*)$.
The only other change in the proof is that $G^+$ gets replaced by $G^+_c$.
\end{proof}

Further we get the following corollary:

\begin{cor}\label{mainthm}
Let $A \in \mathbb{R}^{n \times d}$ be a $\mu$-complex matrix for some bounded $1\leq \mu < \infty$.
There is an algorithm that solves logistic regression in $O(\nnz(A)+ \poly(\mu d \log n) )$ time up to a constant factor with constant probability.
\end{cor}

\section{Experiments}\label{sec:experiments}
\begin{figure*}[ht!]
\begin{center}
\begin{tabular}{ccc}
\includegraphics[width=0.318\linewidth]{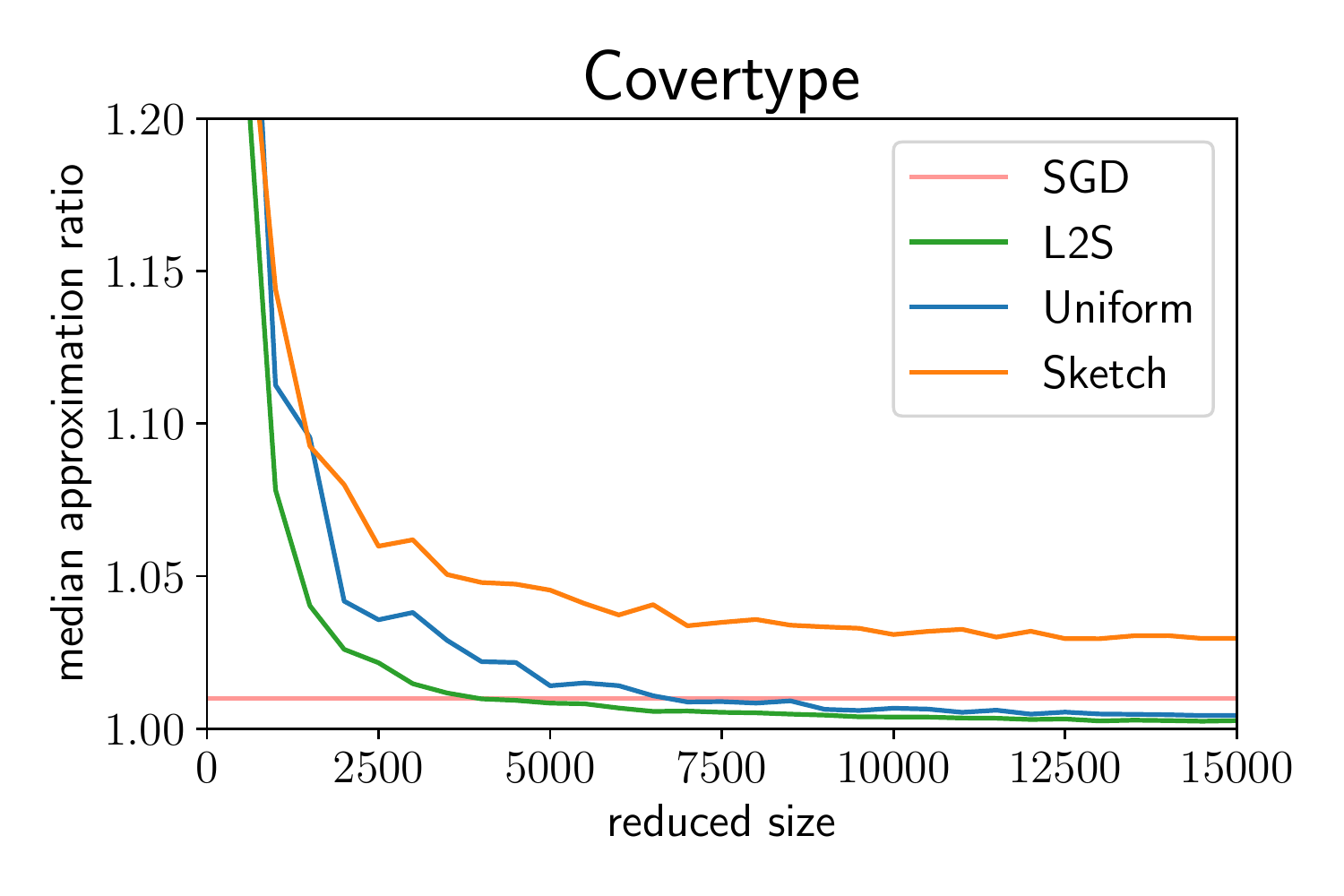}&
\includegraphics[width=0.318\linewidth]{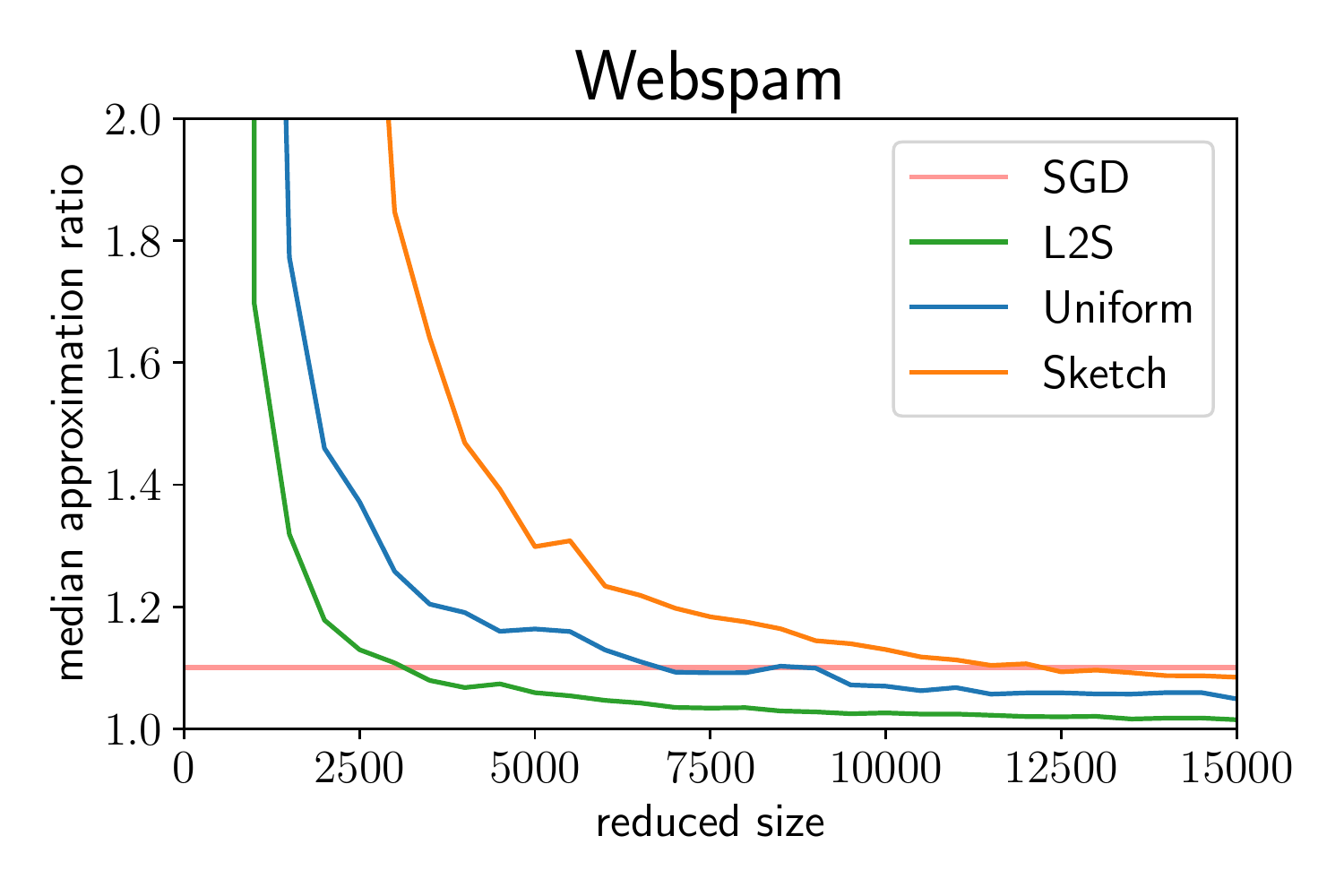}&
\includegraphics[width=0.318\linewidth]{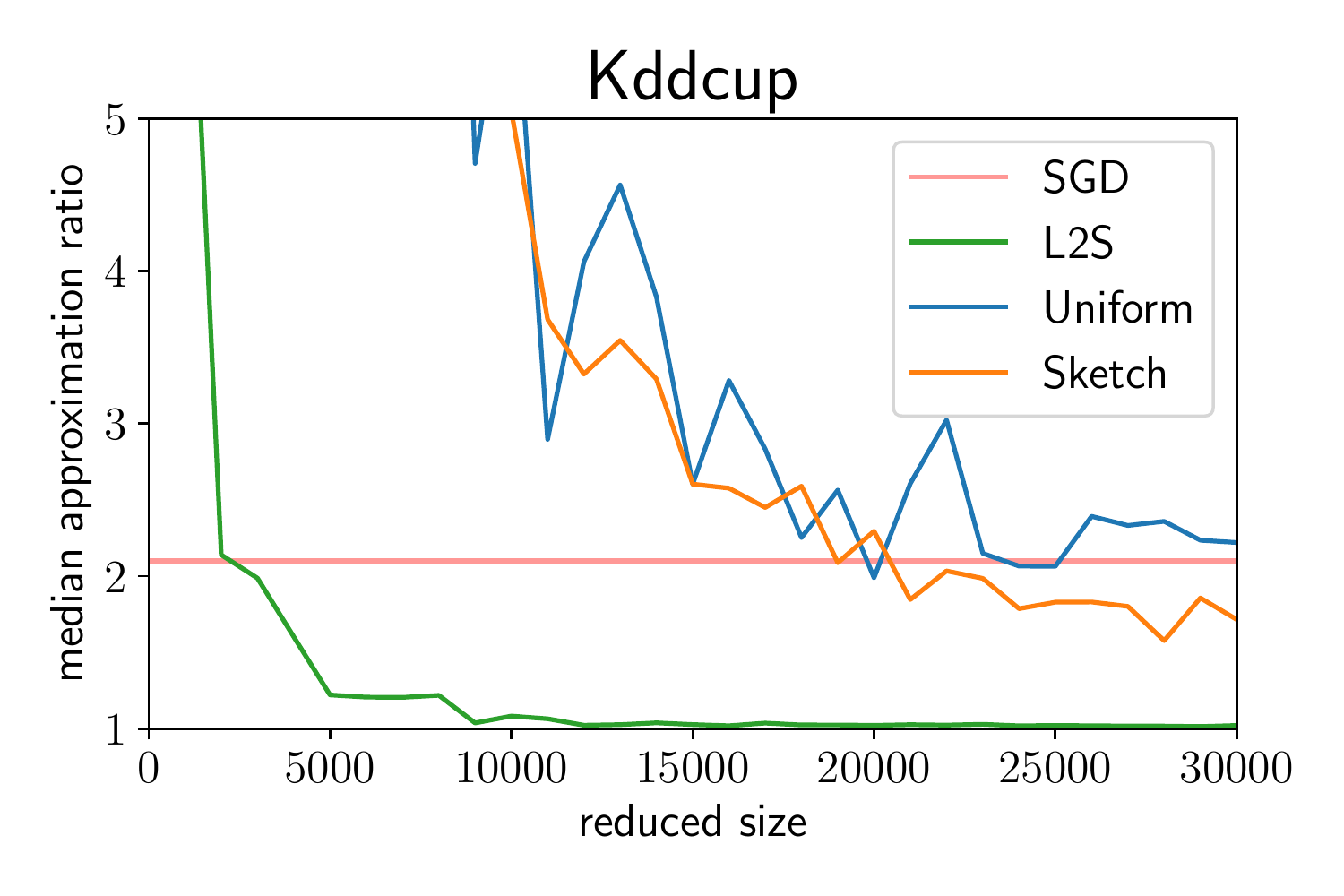}\\
\includegraphics[width=0.318\linewidth]{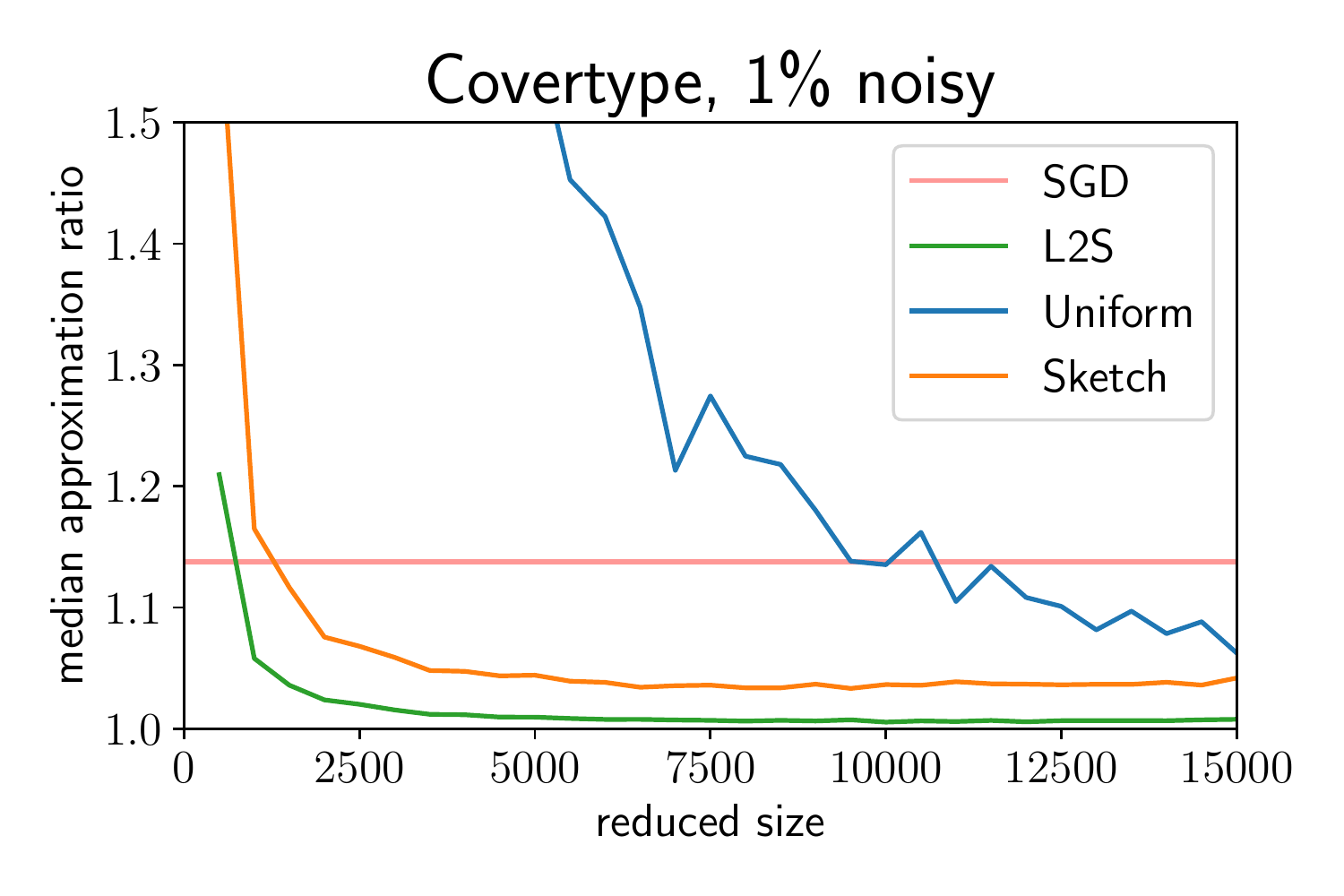}&
\includegraphics[width=0.318\linewidth]{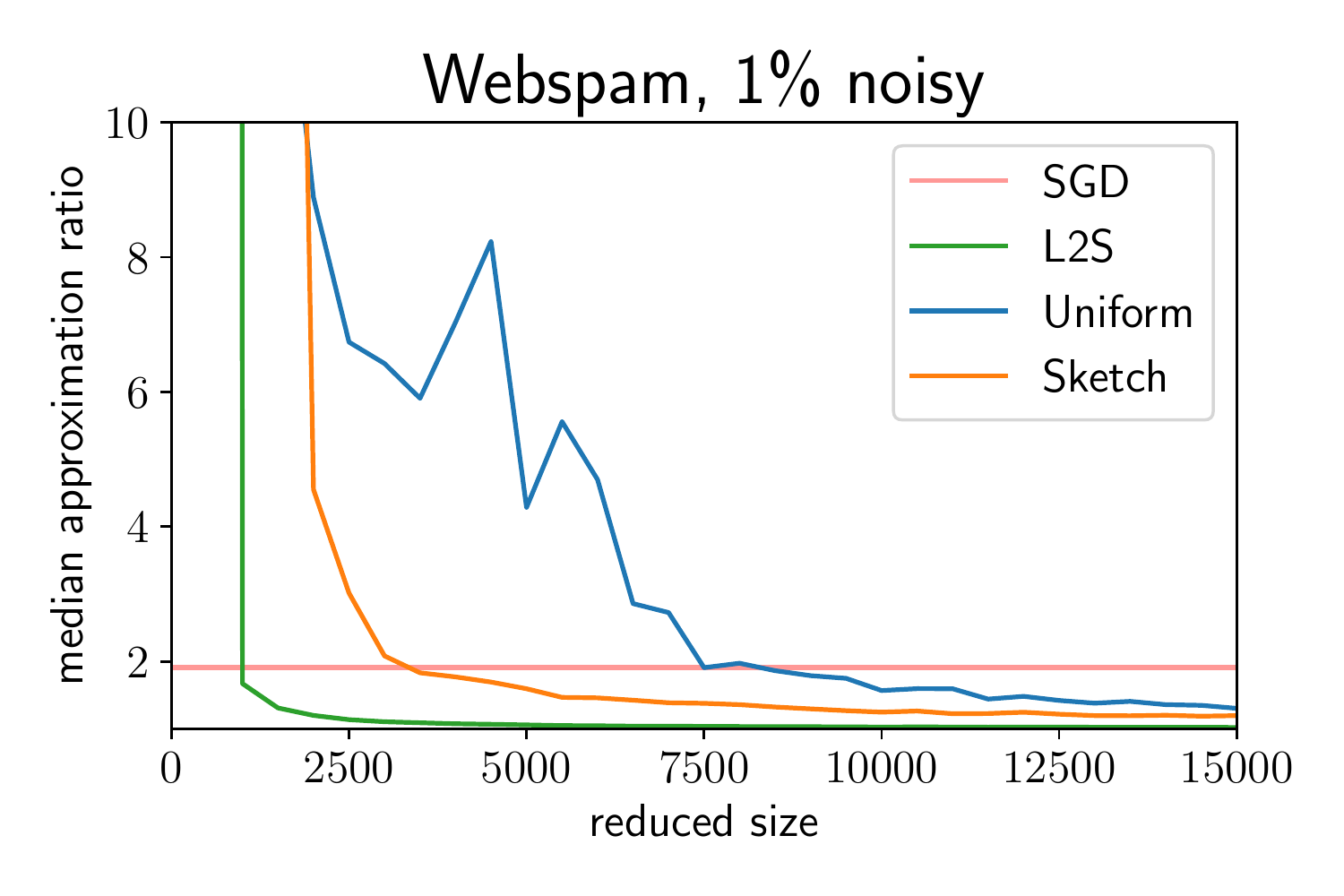}&
\includegraphics[width=0.318\linewidth]{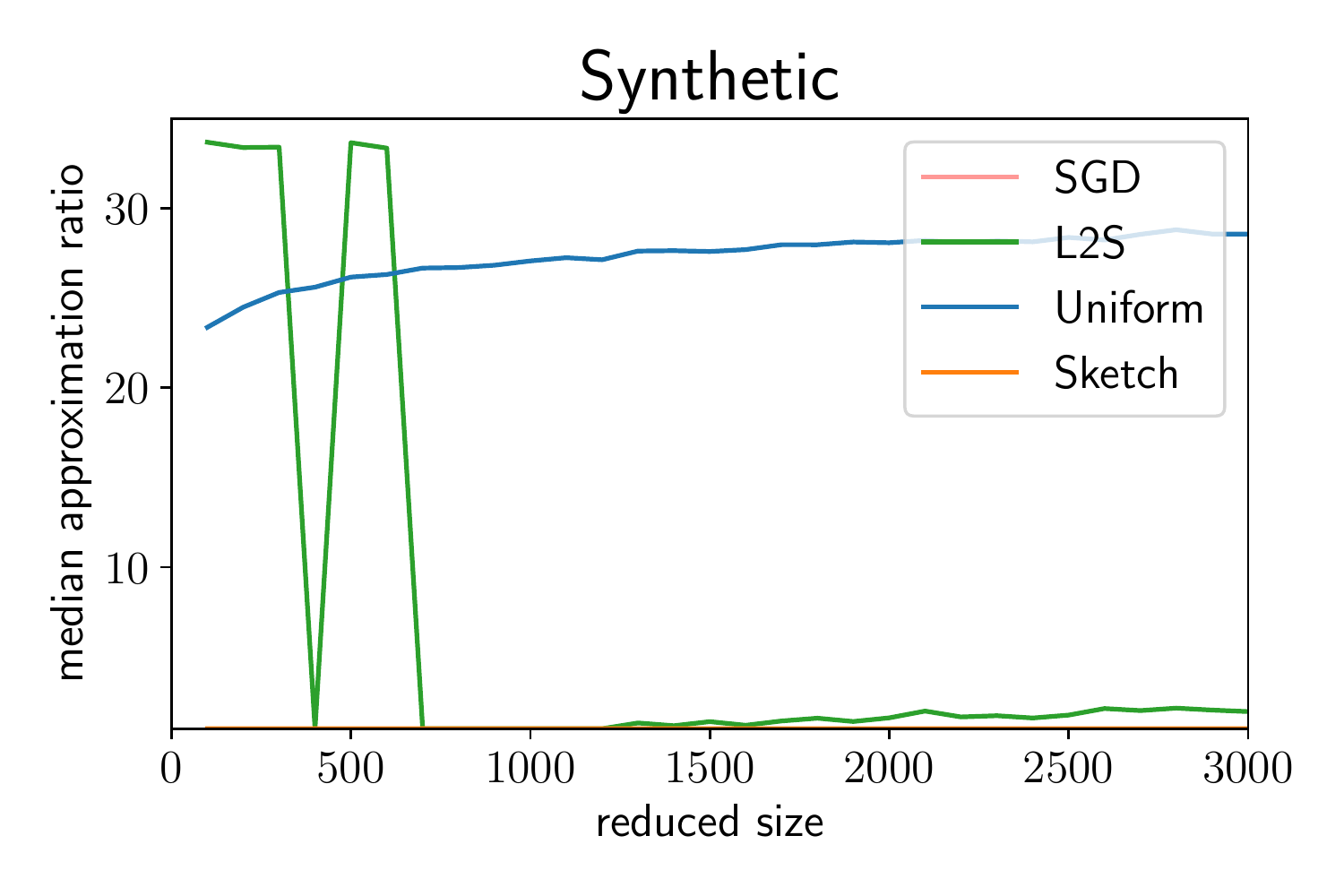}\\
\includegraphics[width=0.318\linewidth]{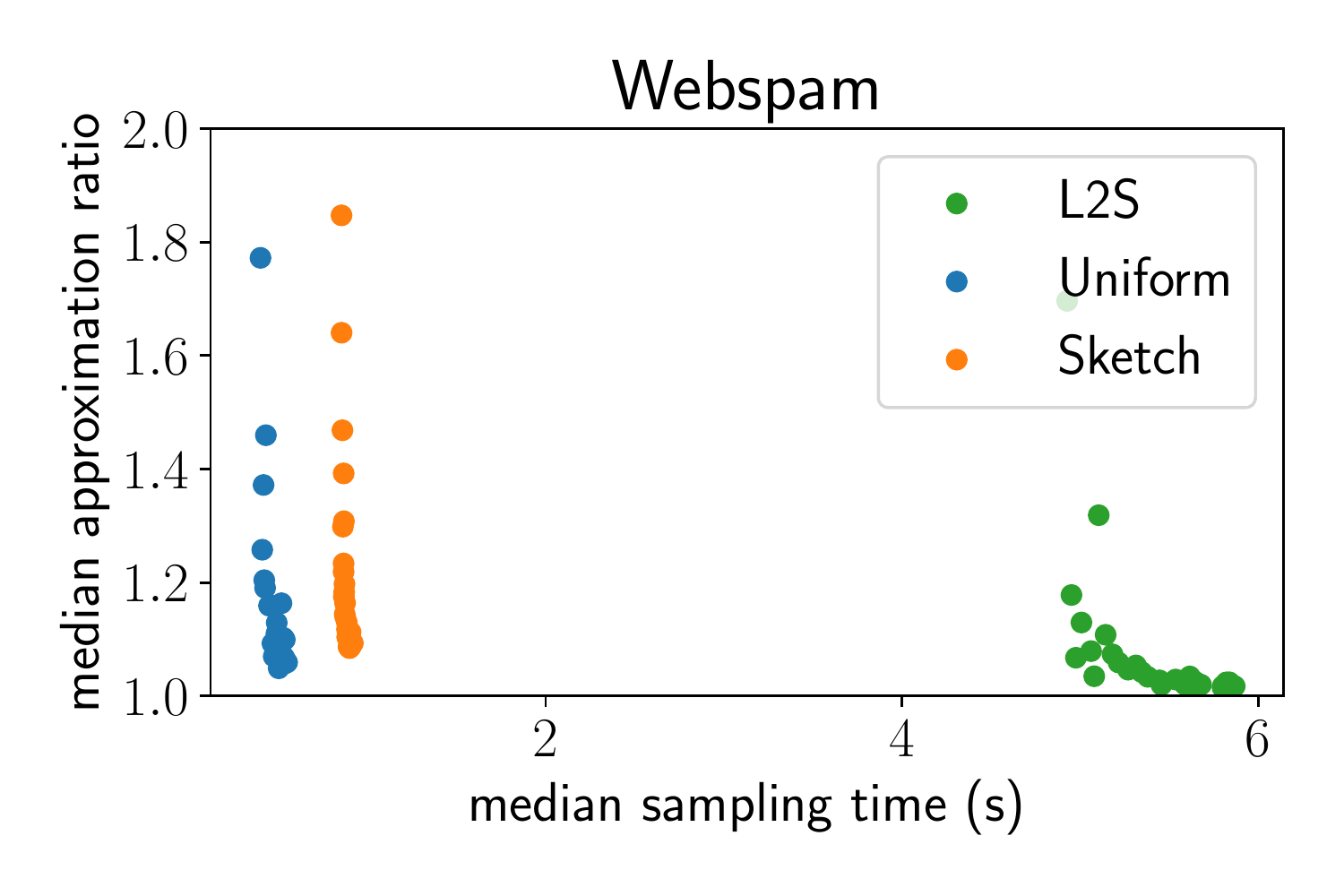}&
\includegraphics[width=0.318\linewidth]{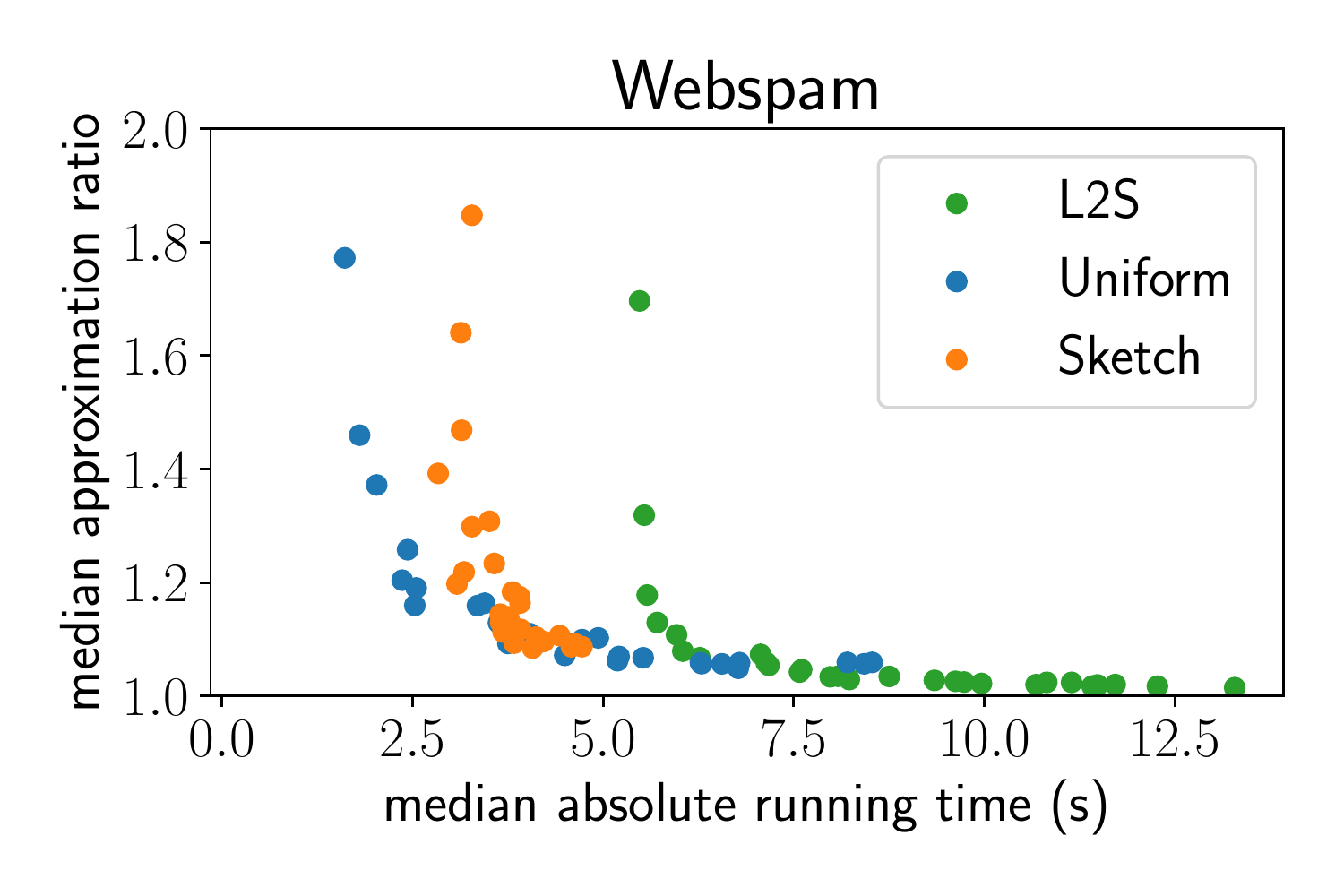}&
\includegraphics[width=0.318\linewidth]{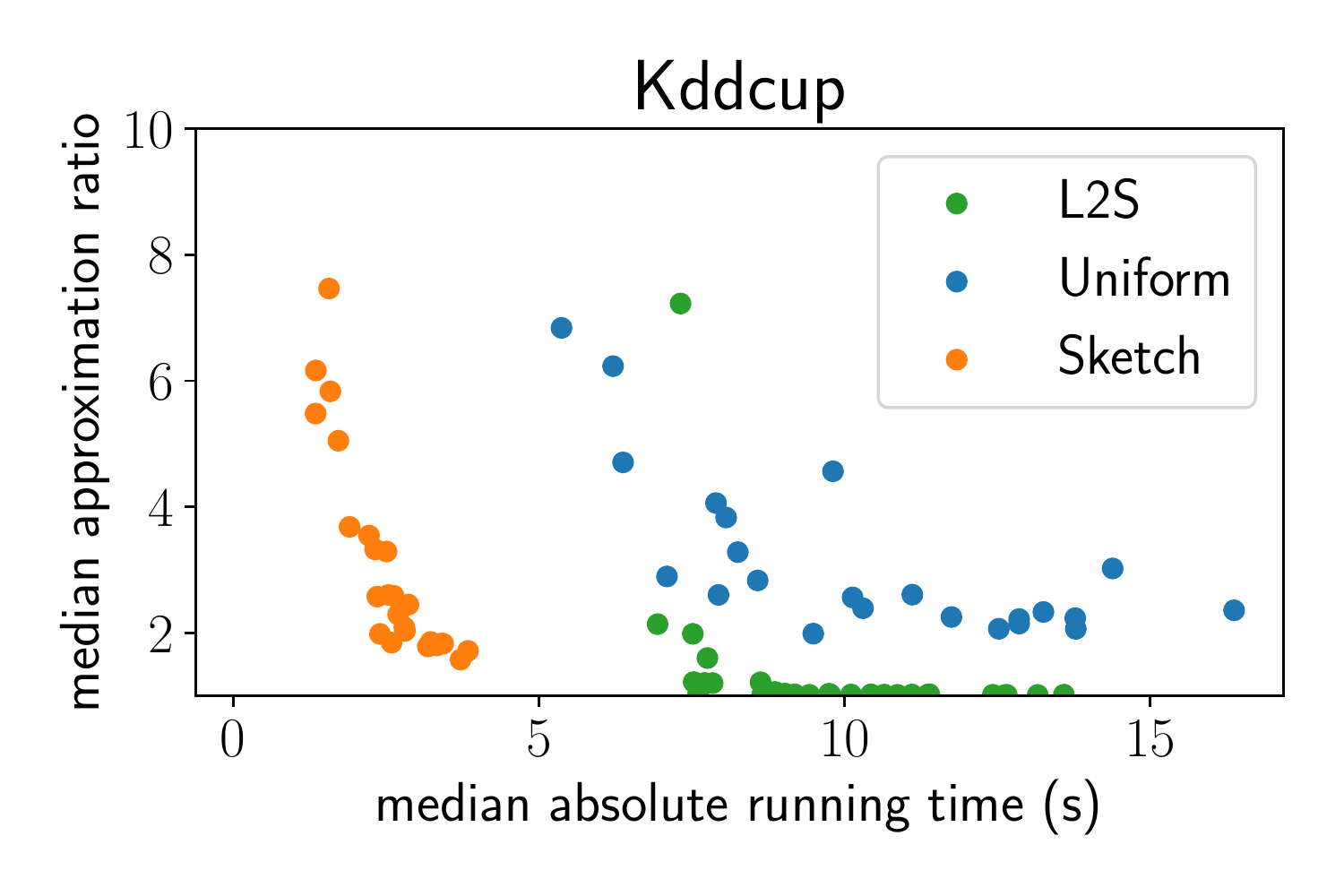}
\end{tabular}
\caption{Comparison of the approximation ratios with and without added noise. Comparison of sketching resp. sampling times vs. accuracy. Comparison of total running times including optimization vs. accuracy.
}
\label{fig:strct_exp}
\end{center}
\end{figure*}
Our results can be reproduced with our open Python implementation available at \url{https://github.com/cxan96/oblivious-sketching-logreg}.
We compare our oblivious \texttt{LogReg}-sketch algorithm with uniform sampling (UNI), stochastic gradient descent (SGD), and the $\ell_2$-leverage score (L2S) coreset from \cite{MunteanuSSW18}. SGD and L2S work only on data presented row-by-row. L2S requires two passes and is not data oblivious. We additionally note that the actual distribution is taking the \emph{square root} of the $\ell_2$-leverage scores as an approximation to $\ell_1$ which is not covered by one-pass online leverage score algorithms given by \cite{CohenMP20,ChhayaC0S20}. We do not compare to the other coresets surveyed in the related work because they rely on regularization, and do not apply to the plain unconstrained logistic regression studied in this paper.
SGD is allowed one entire pass over the data to be comparable to our sketch. Its plot thus displays a flat line showing the final result as a baseline. For all other data reduction methods, the optimization is done with standard optimizers from the scikit learn library\footnote{\url{https://scikit-learn.org/}}.
The covertype and kddcup data sets are loaded automatically by our code from the scikit library, and webspam data is loaded from the LIBSVM data repository\footnote{\url{https://www.csie.ntu.edu.tw/~cjlin/libsvmtools/datasets/}}. Additional details on the size and dimensions of the data sets are in the supplementary, Section \ref{sec:supp:experiments}.
The synthetic data is constructed so that it has $n$ points in one place and two more points are placed in such a way that the leverage score vectors will be $(\frac{1}{n},\ldots,\frac{1}{n},\frac{1}{2},\frac{1}{2})$ for $\ell_1$, $(\frac{1}{\sqrt n},\ldots,\frac{1}{\sqrt n},\frac{1}{2},\frac{1}{2})$ for $\ell_2$ and it is crucial for any reasonable approximation to find those two points. The probability to see one of them is roughly $\frac{1}{n}$ for UNI and SGD, and $\frac{1}{\sqrt n}$ for L2S, but for $\ell_1$-leverage scores 
it will be $1/2$ and thus the points will be heavy hitters and sketched well in separate buckets with constant probability (see Section \ref{sec:contraction}).
The \texttt{LogReg}-sketch uses $h_{\max}+1=3$ levels and one level of uniform sampling. By the Ky Fan argument all but the largest 25\% entries are cut off at each level. The other algorithms were run using their standard parameters.
We varied the target size of all reduction algorithms in thirty equal steps and calculated the approximation ratio
$f(A\tilde x)/f(Ax^*)$ where $x^*$ is the solution returned on the full problem and $\tilde x$ is the solution returned on the reduced version. We repeated each experiment twenty times and displayed the median among all repetitions in Figure \ref{fig:strct_exp}.

For the real data that seem easy to subsample uniformly, we show what happens if we introduce random Gaussian $N(0,10^2)$ noise to 1\% of data to simulate adversarial corruptions; displayed in Figure \ref{fig:strct_exp}. Finally we assess the sampling time as well as the total running time (including the subsequent optimization) vs. the accuracy of our sketch displayed in Figure \ref{fig:strct_exp}. Further plots can be found in the supplementary, Section \ref{sec:supp:experiments}.\\

\textbf{Discussion.}
The overall picture is that \texttt{LogReg}-sketch never performs much worse than its competitors, even for data that is particularly easy to handle for UNI and SGD (see covertype and webspam).
On the kddcup data we see that \texttt{LogReg}-sketch improves slowly with increasing sketch sizes and performs slightly better than UNI. Here L2S clearly performs best. However, we emphasize again that L2S can choose its sampling distribution adaptively to the data and requires two passes in row-order. In contrast \texttt{LogReg}-sketch makes its random choices obliviously to the data and allows single-pass turnstile streaming. The higher flexibility justifies slightly weaker (but still competitive) approximations.
On the synthetic data we see the theoretical claims confirmed. By construction UNI and SGD (not in the plot since its median approximation ratio exceeds 1000) have no chance to give a good approximation on a sublinear sample. L2S starts to converge at about $\Theta(\sqrt n)$ samples. \texttt{LogReg}-sketch has roughly zero error even for very small constant sketch sizes.

When noise is added to corrupt a small number of data points, we see that UNI is not robust and its accuracy deteriorates. In comparison our Sketch (and L2S) are unaffected in this adversarial but natural setting due to their worst case guarantees.

The sketching time of our Sketch is slightly slower but closest to UNI, and the adaptive L2S takes much longer to construct the coreset. When the subsequent optimization is included, we see that the Sketch sometimes becomes even faster than UNI, which indicates that the Sketch produces a summary that is better preconditioned for the optimizer than UNI.

In summary \texttt{LogReg}-sketch provably produces good results that are competitive with the other methods and better than expected from the theoretical $O(1)$ approximation. It is weaker only when compared to the adaptive sampling algorithm that is not applicable in several challenging computational settings motivated in the introduction. We also demonstrated that UNI and SGD have no error guarantees under corruptions and in the worst case where \texttt{LogReg}-sketch even outperforms the adaptive algorithm. In all cases \texttt{LogReg}-sketch performs almost the same or better compared to its competitors and comes with higher flexibility in the aforementioned computational scenarios.

\section{Conclusion}\label{sec:conclusion}
We developed the first data oblivious sketch for a generalized linear model, specifically for logistic regression, which is an important model for classification \cite{ShalevSBD14} and estimation of Bernoulli probabilities \cite{McCullaghN89}. The sketching matrix can be drawn from a data-independent distribution over sparse random matrices which is simple to implement and can be applied to a data matrix $A$ over a turnstile data stream in input-sparsity time. This is important and has advantages over existing coreset constructions when it comes to high-velocity streaming applications and when data is not presented in row-order but in an arbitrary unstructured way. The resulting sketch of polylogarithmic size can be put in any solver for weighted logistic regression and yields an $O(\log n)$-approximation. We also showed how the same sketch can be slightly adapted to give an $O(1)$-approximation. 
Our experiments demonstrate that those sketching techniques are useful and competitive to uniform sampling, SGD, and to state of the art coresets.

\section*{Acknowledgements}
We thank the anonymous reviewers for their valuable comments.
Alexander Munteanu and Simon Omlor were supported by the German Science Foundation (DFG), Collaborative Research Center SFB 876, project C4 and by the Dortmund Data Science Center (DoDSc). D. Woodruff would like to thank NSF grant No. CCF-181584, Office of Naval Research (ONR) grant N00014-18-1-256, and a Simons Investigator Award. 
We thank Christian Peters and Alexander Neuhaus for their help with the experiments.


\bibliographystyle{alpha}
\bibliography{references}

\clearpage

\appendix
\section{Preliminaries}

Bernstein's inequality will be used multiple times:
\begin{pro}
	\label{thm:bernstein}[Bernstein's Inequality]\cite{Maurer03}
Let $\{X_i\}$ be independent random variables with $\mathbb{E}(X_i^2) < \infty$ and $X_i\geq 0$. Set $X=\sum_i X_i$ and $\lambda >0$. Then,
\begin{equation*}
Pr[X \leq \mathbb{E}(X)-\lambda ]\leq \exp\left( \frac{-\lambda^2}{2 \sum_i \mathbb{E}(X_i^2)} \right) .
\end{equation*}
If $X_i -  \mathbb{E}(X_i)\leq \Delta $ for all $i$, then with $\sigma_i^2=\mathbb{E}(X_i^2)- \mathbb{E}(X_i)^2$ we have
\begin{equation*}
Pr[X \geq \mathbb{E}(X)+\lambda ]\leq \exp\left( \frac{-\lambda^2}{2 \sum_i \sigma_i^2 + \lambda \Delta} \right).
\end{equation*}
\end{pro}

\section{Omitted details from Section \ref{secapprox}}\label{sec:supp:approx}
For technical reasons we make the following assumptions on the parameters:

\begin{ass}\label{Ass3.1}
We assume that the following inequalities hold for $N, m, b, \delta$:
\begin{align*}
& m\varepsilon^2 \geq \max\left\lbrace -4\ln(\delta),~ 3\ln\left(\beta m \log_2\left( \frac{m}{\varepsilon}\right) \right),
 2\ln\left(\log_2 \left(\frac{n}{m\varepsilon}\right)+1\right), 2d\ln\left(1+\frac{n}{m\varepsilon}\right)-\ln(\delta) \right\rbrace\\
& b \geq \max\left\lbrace m,~ \delta^{-1} \right\rbrace \\
& N = \Omega (bm^2 d^2 \varepsilon^{-1}\delta^{-1})
\end{align*}
\end{ass}

\subsection{Contraction bounds}\label{sec:contraction}

We set $Q_1=\{q \in \{1 ,\dots, q_{\max}\} ~|~ |W_q^+|\geq \beta m \wedge \Vert W_q^+\Vert_1 \geq \frac{\varepsilon}{q_{\max}} \}$ and $Q_2=\{q \leq \log_2(\frac{m}{\varepsilon})\}\setminus Q_1$.
We set $Q^*=Q_1 \cup Q_2$ to be the set of important weight classes.
The following lemma shows that the weight of the remaining weight classes is negligible:

\begin{lem}\label{Lem3.7}
It holds that $\sum_{q  \in Q^{*}} \Vert W_q^+ \Vert_1 \geq (1-5\varepsilon')G^+(z)$.
\end{lem}
{
\begin{proof}
The total weight $W$ of those classes with $q \notin Q_1\cup Q_2$ is at most
\begin{align*}
 W &\leq\sum_{q\leq q_{\max}, q \notin Q^*} \Vert W_q^+\Vert_1 ~+~ \sum_{q>q_{\max}}\Vert W_q^+\Vert_1
 \leq \frac{\varepsilon}{m} \beta m \sum_{q=0}^{\infty}2^{-q} ~+~ 2^{-q_{\max}}n
 =4\varepsilon+ \frac{\varepsilon}{n}n=5\varepsilon 
\end{align*}
as $\beta \leq 2$.
Recall that $G^+(z)\geq \frac{1}{\mu'}$.
Combining these two facts gives us
\begin{align*}
\sum_{q  \in Q^{*}} \Vert W_q^+ \Vert_1 
&=G^+(z)-\sum_{q  \notin Q^{*}} \Vert W_q^+ \Vert_1  
\geq G^+(z)-5\varepsilon
\geq G^+(z)- 5G^+(z)\mu' \varepsilon
= G^+(z)-5\varepsilon'G^+(z).
\end{align*}
\end{proof}
}

In the following we show that the important weight classes have at least the same contribution, up to a small error, after sketching.
First we consider the weight classes with $q \in Q_2$, where the individual entries have a notable contribution themselves.
Then we consider the weight classes with $q \in Q_1$ which consist of a large number of entries.
In both cases we will show that for each important weight class $W_q^+$ there exists a subset $W_q^* \subset W_q^+$ with $\Vert W_q^* \Vert_1 \geq (1-7\varepsilon) \frac{\Vert W_q^+ \Vert_1}{ b^{h(q)}\beta}$ and where each entry $z_p \in W_q^+$ is much larger than the sum of all other entries in its bucket.






\textbf{Heavy-hitters.} This section is dedicated to showing that the large entries of $z$ are well separated among the buckets and that they contribute about the same value after sketching.
For $A \in \mathbb{R}^{n \times d}$, let $u \in \mathbb{R}^n$ denote the $\ell_1$-leverage score vector of $A$, i.e.,
\[ u_i=\max_{x \in \mathbb{R}^d\setminus\{0\}} \frac{|(Ax)_i|}{\Vert Ax\Vert_1}.\]

We start by showing that the rows of $A$ with the largest $\ell_1$-leverage scores are well separated and that only coordinates $z_p$ with high $\ell_1$-leverage can be large coordinates of $z$. To this end we need two lemmas:

\begin{lem}\label{Lem3.4}\cite{ClarksonW15}
For $N_1, N_2$ with $N_2 \geq N_1$ and with $N_1N_2 \leq \kappa N$, for $\kappa \in (0, 1/2)$, let $Y_1$ and $Y_2$ denote the sets of indices of the $N_1$ and $N_2$ coordinates with the largest leverage scores, so that $Y_1 \subset Y_2$.
Then with probability at least $1-2\kappa$ each member of $Y_1$ is in a bucket containing no other member of $Y_2$.
\end{lem}

\begin{lem}\label{Lem3.5}
If $u_p$ is the $k$-th largest coordinate of $u$, then for $z$ in the subspace spanned by the columns of A 
it holds that $|z_p|\leq \frac{d}{k}G(z)$.
\end{lem}
\begin{proof}
By \cite{DasguptaDHKM09} there exists a so-called \emph{Auerbach} basis $Q$ of $A$ with the following properties. 
It holds that $G(Qx)=\|Qx\|_1 \geq \Vert x \Vert_\infty$ for all $x \in \mathbb{R}^d$ and $\sum_{ij} |Q_{ij}|\leq d$.
Note that by a change of basis
\begin{align*}
u_i&=\max_{x \in \mathbb{R}^d\setminus\{0\}} \frac{|(Ax)_i|}{\Vert Ax\Vert_1} =\max_{x \in \mathbb{R}^d\setminus\{0\}} \frac{|(Qx)_i|}{\Vert Qx\Vert_1} .
\end{align*}
Thus $|z_i|=|Q_i x| \leq \|Q_i\|_1 \|x\|_\infty \leq \|Q_i\|_1 \|Qx\|_1$ and it follows that $\sum_i u_i \leq \sum_i \|Q_i\|_1 = \sum_{ij} |Q_{ij}| \leq d$.
Consequently the $k$-th largest coordinate of $z$ can be at most $|z_p| \leq u_p G(Qx) \leq \frac{d}{k} G(Qx)=\frac{d}{k} G(z)$.
\end{proof}

We apply Lemma \ref{Lem3.4} and Lemma \ref{Lem3.5} as follows: set $N_1= \frac{d \beta m}{\varepsilon}$ and $N_2 = d \beta m^2$, and let $Y_1$ (resp. $Y_2$) be the set of coordinates with the $N_1$ (resp. $N_2$) largest leverage scores.
We denote by $\mathcal{E}_1$ the event that all coordinates in $Y_1$ are in a bucket with no other member of $Y_2$.
By Lemma \ref{Lem3.4} and Assumption \ref{Ass3.1}, $\mathcal{E}_1$ holds with probability at least $1-\delta$.
Then by Lemma \ref{Lem3.5}, for any entry $z_p\geq v_1:= \frac{\varepsilon}{\beta m}$ we have $p \in Y_1$ and for any entry $p \in Y_2$ we have $z_p < \frac{1}{\beta m^2 }= \frac{\varepsilon}{\beta m } \cdot \frac{\varepsilon}{m \varepsilon^2}=\frac{\varepsilon}{m \varepsilon^2}v_1$.
It remains to show that the remaining entries in the buckets containing a heavy hitter only have a small contribution.
To show this we use Bernstein's inequality.
For a coordinate $p \in [n]$ we denote by $B_p$ the bucket that contains $p$.

\begin{lem}\label{Lem3.10}
Assume $\mathcal{E}_1$ holds.
Then with failure probability at most $e^{-m\varepsilon^2}$ for any $p$ with $z_p \geq v_1$ and $z_p \in W_q$ for some $q \in Q_2$, we have $\Vert B_p\setminus \{p\}\Vert_1 \leq 3\varepsilon | z_p |$.
\end{lem}
{
\begin{proof}
Let $p$ with $z_p \geq v_1$.
Then by Lemma \ref{Lem3.5} it holds that $p \in Y_1$.
For each $i \in [n]\setminus \{p\}$ we define a random variable $X_i$ by $X_i=z_i$ if $i$ is put in bucket $B_p$ and $X_i=0$ otherwise.
By our assumption, $X_i=0$ if $i \in Y_2$.
Otherwise we have $P(X_i \neq 0) \leq \frac{1}{\beta N} $ since the probability that any coordinate is put in any bucket is at most $ \frac{1}{\beta N}$.
Further we have $\mathbb{E}(\Vert B_p\setminus \{p\}\Vert_1)\leq \frac{1}{\beta N} $ since $G(z)=1$.
For the variance we have 
\begin{align*}
\sigma_i^2 &= (z_i- \mathbb{E}(X_i))^2 P(X_i=z_i) + \mathbb{E}(X_i)^2P(X_i=0)\leq z_i^2 \frac{1}{\beta N}+ \frac{z_i^2}{(\beta N)^2}< \frac{2z_i^2}{\beta N}
\end{align*}
Since $z_{i} \leq \frac{\varepsilon}{m \varepsilon^2}v_1$ for $i \notin Y_2$ this implies 
\begin{align*}
\sum_{i \in [n]} \sigma_i^2 
=\sum_{i \notin Y_2} \sigma_i^2
<  \sum_{i \notin Y_2}\frac{2z_i^2}{\beta N}
<\frac{\varepsilon}{m \varepsilon^2\beta N}v_1 \sum_{i \in [n]} z_i 
=\frac{\varepsilon}{m \varepsilon^2\beta N}v_1 
\leq \frac{\varepsilon^2 }{\beta^2 m^3 q_{\max}^2} 
\end{align*}
since $N\geq  \varepsilon^{-1}$.
Thus, applying Bernstein's inequality with $\lambda=\varepsilon v_1$ and $\Delta=\frac{\varepsilon}{m \varepsilon^2}v_1 =\frac{\lambda}{m\varepsilon^2}$ we get
\begin{align*}
Pr[X \geq \frac{1}{\beta N}+ 2\lambda ]&\leq \exp\left( \frac{-4\lambda^2}{2 \sum_i \sigma_i^2 + \lambda \Delta} \right) \\
&\leq \exp\left( \frac{-4\lambda^2}{2 \lambda^2 /(m\varepsilon^2) + \lambda^2/(m\varepsilon^2)} \right)\\ 
&\leq e^{-4m\varepsilon^2/3} . 
\end{align*}
Using that $\frac{1}{\beta N}+ 2\lambda \leq 3 \varepsilon z_p$ and using the union bound for at most $\beta m \log_2\left(\frac{m}{\varepsilon}\right)$ coordinates $z_p$ with $z_p \in W_q^+$ for some $q \in Q_2$ concludes the proof, as the total failure probability $P$ is bounded by
\begin{align*}
P&=\beta m \log_2\left(\frac{m}{\varepsilon}\right) \exp(-4m\varepsilon^2/3)\\ 
&= \exp\left(\ln\left(\beta m \log_2\left(\frac{m}{\varepsilon}\right)\right)-4m\varepsilon^2/3\right)\\
&\leq \exp(-m\varepsilon^2)
\end{align*}
by Assumption \ref{Ass3.1}.
\end{proof}
}

\textbf{Other important weight classes.} Let $q \in Q_1$.
Then we have  $ |W_q^+| \geq \beta m$.
We show that we can find a set of representatives for $W_q$ where each representative $z_p$ is in a bucket with no other entry larger than $\frac{|z_p|}{ \varepsilon m}$.
We set $ Y_q=\{ p \in [n] ~|~ h_p=h(q) \text{ and } |z_p| \geq \frac{2^{q-1}}{m\varepsilon}  \}$.

\begin{lem}\label{Lem3.2}
Let $q \in Q_1$ and $h=h(q)$.
Then with failure probability at most $\exp(-m \varepsilon^2)$ there exists a subset $W_q^* \subset W_q^+$ such that $\Vert W_q^* \Vert_1 \geq (1-7\varepsilon) \frac{\Vert W_q^+ \Vert_1}{ b^{h}\beta}$ and every $z_p \in W_q^*$ is in a bucket with no other element of $Y_q$.
\end{lem}
{
\begin{proof}
For $z_i \in W_q$ define $X_i=1$ if $h_i=h$ and $X_i=0$ otherwise, and set $X=\sum X_i$.
The expected number of entries from $W_q^+$ at level $h$ is $\mathbb{E} (X)=\frac{|W_q^+|}{\beta b^h}\in [m, bm)$ by definition of $h=h(q)$.
Using Bernstein's inequality we get that with failure probability at most $P_1=\exp\left(\frac{-(2\varepsilon \mathbb{E} (X))^2}{2\mathbb{E} (X)}\right)\leq \exp(-2\varepsilon^2 m)$, there are at least $\frac{|W_q^+|}{\beta b^h}(1-2\varepsilon)$ entries of $W_q^+$ at level $h$.
We denote a uniform random subset of size $y=\frac{|W_q^+|}{\beta b^h}(1-2\varepsilon)$ of such entries by $W_q'$.

Since $q \in Q_1$ we have $|W_q^+|\geq \frac{\varepsilon/q_{\max}}{2^{1-q}}$ and thus also $h= \lfloor \log_b(\frac{|W_q^+|}{\beta m}) \rfloor\geq \log_b(\frac{\varepsilon 2^{q-1}}{ \beta mq_{\max}})-1$.
Since $G(z)=1$ there are at most $2^{q-1}m\varepsilon $ entries larger or equal to $ \frac{1  }{m\varepsilon2^{q-1}}$.
The expected number of entries from $Y_q$ at level $h$ is thus bounded by $E=\frac{2^{q-1}m\varepsilon}{\beta  b^h} \leq q_{\max} m^2 b $.
To show that the number is not much larger than that, we define independent random variables $X_i=1$ if $h_i=h$ and $X_i=0$ otherwise, for $i \in Y_q$.
The variance is bounded by
\[ \sum_{i \in Y_q}\sigma_i^2 \leq \sum_{i \in Y_q} \frac{1}{\beta b^h}= E. \]
Thus using Bernstein's inequality we get
\begin{align*}
 P_2=P(|Y_q| \geq E+E) &\leq \exp\left( \frac{-E^2}{2E+E} \right) \\
 &= \exp\left( \frac{-E}{3} \right)\\
 &\leq \exp(- 2m  \varepsilon^{2}).
\end{align*}

Now for $z_i \in W_q'$ consider the random variable $X_i=1$ if $z_i$ is put into a bucket with another entry of $Y_q$, and $X_i=0$ otherwise. Set $X=\sum X_i$.
We have $P(X_i=1)<\frac{2E}{N}$ and $\mathbb{E}(X)\leq \frac{2Ey}{N}$.
The variance is bounded by $\sigma_i^2=\frac{2E}{N}+\left(\frac{2E}{N} \right)^2\leq \frac{4E}{N}$.
Thus up to failure probability
\begin{align*}
P_3=P(X \geq \frac{2Ey}{N}+\varepsilon y)&\leq \exp\left( -\frac{2y^2\varepsilon^2}{8Ey/N+y\varepsilon} \right)\\
&\leq\exp(-m(1-2\varepsilon)\varepsilon/2)\\
&\leq\exp(-2m\varepsilon^2)
\end{align*}
we have that at most $2(\frac{2Ey}{N}+\varepsilon y)\leq 4\varepsilon y$ entries in $W_q'$ are in a bucket with another entry of $Y_q$.
(Here the factor of $2$ comes from the possibility that $z_i$ is placed in a bucket with another entry of $W_q'$ which was put into the bucket before $z_i$.) 
We denote by $W_q^*$ the subset of $W_q'$ that is in a bucket with no other entry of $Y_q$, and set $y'=|W_q^*|$.
Note that with with failure probability at most $P_3$, we have $y'\geq \frac{|W_q^+|}{\beta b^h}(1-6\varepsilon)$.

Now for $z_i \in W_q^*$ consider the random variable $X_i=z_i$ if $z_i\in W_q^*$, and $X_i=0$ otherwise.
Again using Bernstein's inequality we get
\begin{align*}
P_4&=P(\Vert W_q^* \Vert_1 \leq  \frac{y'}{|W_q^+|}  \Vert W_q^* \Vert_1 - 2y'\varepsilon 2^{-q})\\
&\leq \exp\left(\frac{-(2\varepsilon y' 2^{-q})^2}{2y' 2^{-2q}}\right)\\
&= \exp\left( -2y' \varepsilon^2 \right) \leq \exp(-3m \varepsilon^2/2).
\end{align*}
Since $\varepsilon y' 2^{-q} \leq \varepsilon \frac{\Vert W_q^+ \Vert_1}{ b^{h}\beta}$ this shows that $W_q^*$ has the desired properties with failure probability at most $P_1+P_2+P_3+P_4\leq 4\exp(-2m \varepsilon^2/2)\leq \exp(-m \varepsilon^2)$.
\end{proof}
}

\begin{lem}\label{Lem3.9}
Let $q \in Q_1$ and assume $W_q^*$ from Lemma \ref{Lem3.2} exists. Then with failure probability at most $e^{-m\varepsilon^2}$, for any $z_p \in W_q^*$ we have $\Vert B_p\setminus \{p\}\Vert_1 \leq 3\varepsilon | z_p |$.
\end{lem}

The proof of Lemma \ref{Lem3.9} is similar to the proof of Lemma \ref{Lem3.10}.

\textbf{Contribution of important weight classes.} By combining Lemma \ref{Lem3.7}, Lemma \ref{Lem3.9}, and Lemma \ref{Lem3.10}, we can prove Theorem \ref{Thm3.1}:

{
\begin{proof}[Proof of Theorem \ref{Thm3.1}]
By assumption, $\mathcal{E}$ holds. Further, by a union bound, with failure probability at most 
\begin{align*}
P&=q_{\max}\exp(-m \varepsilon^2)+(q_{\max}+1)\exp(- m\varepsilon^2)\\
&\leq \exp(\ln(2q_{\max}+1)-m \varepsilon^2)\\
&\leq \exp(- m\varepsilon^2/2)
\end{align*}
for each important weight class there exists $W_q^*$ as in Lemma \ref{Lem3.2}, and the events of Lemmas \ref{Lem3.10} and \ref{Lem3.9} hold.
For $q \leq \log_b\left(\frac{\varepsilon}{\beta m q_{\max}}\right)$ we set $W_q^*=W_q^+$.
We have
\begin{align*}
G^+(Sz) & \geq \sum_{q \in Q^*, z_p\in W_q^*} (1-3\varepsilon) z_p b^{h_p}\beta\\
&\geq \sum_{q \in Q^*}(1-3\varepsilon)\Vert W_q^* \Vert_1 b^{h_p}\beta\\
&\geq \sum_{q \in Q^*}(1-3\varepsilon)(1-7\varepsilon)  \frac{\Vert W_q^+ \Vert_1}{ b^{h(q)}\beta}b^{h(q)}\beta\\
&\geq (1-10\varepsilon)(1-5\varepsilon')G^+(z)\\
&\geq (1-60\varepsilon')G^+(z).
\end{align*}
The first inequality follows by Lemma \ref{Lem3.9} and Lemma \ref{Lem3.10}. The second follows by Lemma \ref{Lem3.2} and the third one by Lemma \ref{Lem3.7}.
\end{proof}
}

\subsection{Net argument}\label{sec:netargument}

Next we show that for all $z$ we have $G^+(Sz)\geq (1-\varepsilon')G^+(z)$
with high probability.
We need the following lemma:

\begin{lem}\label{netlem}
Let $z, e \in \mathbb{R}^n$ with $G(e)\leq \frac{\varepsilon}{b^{h_{\max}}} G(z)$.
Then it holds that $G^+(z+e)= (1 \pm \varepsilon')G^+(z)$ and $G^+(S(z+e))=(1 \pm \varepsilon')G^+(Sz)$.
Further, if $G^+(Sz)\geq (1-\varepsilon')G^+(z)$, then $G^+(S(z+e)) \geq (1-4\varepsilon') G^+(z+e)$.
\end{lem}
{
\begin{proof}
A simple case distinction shows that for all $v, v' \in \mathbb{R}$ it holds that $|\max\{v', 0\}-\max\{v, 0\}|\leq |v'-v|$:
If both $v$ and $v'$ have the same sign then either $|\max\{v', 0\}-\max\{v, 0\}|=0$ or $|\max\{v', 0\}-\max\{v, 0\}|= |v'-v|$, and if $v$ and $v'$ have different signs then $|v'-v|=|v'|+|v|\geq |\max\{v', 0\}-\max\{v, 0\}|$.
Thus we have
\begin{align*}
&~|G^+(z+e) - G^+(z) | \\
& \leq \sum\limits_{i\in[n]} |\max\{z_i+e_i, 0\}-\max\{z_i, 0\}| \\
&\leq G(e) \leq \varepsilon G(z) \\
&\leq \frac{\varepsilon'}{\mu'}G(z)\leq \varepsilon'G^+(z).
\end{align*}
It holds that $G(Se)\leq b^{h_{\max}}G(e)$ since all entries in $S$ are bounded by $b^{h_{\max}}$ and each entry of $e$ is multiplied by exactly one non-zero entry of $S$.
Hence
\begin{align*}
|G^+(S(z+e)) - G^+(Sz) | &\leq G(Se)\\
&\leq b^{h_{\max}}G(e)\\
&\leq \frac{\varepsilon'}{\mu'}G(z)\leq \varepsilon'G^+(z).
\end{align*}
Finally if $G^+(Sz)\geq (1-\varepsilon')G^+(z)$ then by combining the previous inequalities we get
\begin{align*}
~ G^+(S(z+e)) 
&\geq G^+(Sz)-\varepsilon'G^+(z)\\
&\geq (1-\varepsilon') G^+(z)-\varepsilon'G^+(z)\\
&\geq (1-\varepsilon') (G^+(z+e) - \varepsilon'G^+(z))-\varepsilon'G^+(z)\\
&\geq (1-\varepsilon') G^+(z+e) -2\varepsilon' G^+(z)\\
&\geq (1-\varepsilon') G^+(z+e) -2\varepsilon' G^+(z+e)/(1-\varepsilon')\\
&\geq (1-4\varepsilon') G^+(z+e)\,.
\end{align*}
\end{proof}
}

Now we are ready to prove Theorem \ref{netthm}.

{
\begin{proof}[Proof of Theorem \ref{netthm}]
With failure probability at most $\delta$, we can assume that $\mathcal{E}$ holds.
Since $G^+(az)=a(G^+(z))$ for all $z \in \mathbb{R}^d$ and $a \in \mathbb{R}_{\geq 0}$, it suffices to show that $G^+( Sz )\geq (1-\varepsilon')G^+(z)$ holds for all $z \in \mathbb{R}^d$ with $G(z)=1$.
We set $M=\lceil\frac{b^{h_{\max}}}{\varepsilon} \rceil$.
Consider the set $ N_\varepsilon=\{(n_1 , \dots , n_d)\frac{1}{M} ~|~ n_1 + \dots + n_d= M\}$.
The set $ N_\varepsilon$ consists of at most $\left(\left(1+\frac{b^{h_{\max}}}{\varepsilon}\right)^d \right)=\left(\exp\left(d\ln\left(1+\frac{n}{m\varepsilon}\right)\right)\right)$ elements as $h_{\max}=\ln_b\left(\frac{n}{m}\right)$.
By Theorem \ref{Thm3.1} and a union bound we have that $G^+( Sz )\geq (1-60\varepsilon')G^+(z)$ holds for all $z \in N_\varepsilon$ with failure probability at most $\exp\left(d\ln\left(1+\frac{n}{m\varepsilon}\right)\right) \exp(-m\varepsilon^2/2)<\delta$ since $d\ln(1+\frac{n}{m\varepsilon})-m\varepsilon^2/2\leq\ln(\delta)$ by Assumption \ref{Ass3.1}.
For each $z$ with $G(z)=1$ there exists $z' \in N_\varepsilon$ such that $ G(z'-z)\leq \frac{\varepsilon}{b^{h_{\max}}}$.
Thus we can apply Lemma \ref{netlem} and get $G^+(Sz)\geq(1-240\varepsilon')G^+(z) $.
The total failure probability is at most $2\delta$ using the union bound.
\end{proof}
}

\subsection{Dilation bounds}\label{sec:dilation}
Our first dilation bound is very simple but yields only an $h_{\max}=O(\log n)$ approximation.

\begin{lem}\label{Lem3.12}
We have $\mathbb{E}(G^+(Sz) )\leq h_{\max}G^+(z)$.
\end{lem}
{
\begin{proof}
The expected contribution of $z_i$ is less than $0$ if $z_i <0$.
If $z_i \geq 0$ then the expected contribution is upper bounded by $\sum_{h=0}^{h_{\max}} \frac{1}{b^{h}\beta} b^h\beta z_i = \sum_{h=0}^{h_{\max}}z_i=h_{\max}z_i$ with equality if and only if all $z_i$ are greater or equal to zero.
Thus
\[\mathbb{E}(G^+(Sz) )\leq \sum_{z_i\geq 0} z_i h_{\max} = h_{\max}G^+(z). \]
\end{proof}
}

We can achieve a better constant dilation by considering only the top buckets at each level.
More precisely set $K=\beta m \log(\frac{m }{\varepsilon})+\beta mb \log_2(b\varepsilon^{-1}) = {O}(mb\log_2(b\varepsilon^{-1}))$.
We define
\[ G^+_c( Sz ) := \sum_h \beta b^h \sum_{i \in [K]}G^+(L_{h, i}) \]
where $L_{h,i}$ denotes the level $h$ bucket with the $i$-th largest sum of entries among all level $h$ buckets. It is important here to take the buckets with the largest contributions to preserve the convexity of the objective as pointed out in \cite{ClarksonW15}, since the resulting function is related to the Ky Fan norm and is thus convex.
The proof of the bounded contraction of $G^+_c( Sz )$, Theorem \ref{Thm3.2}, only requires lower bounds on $G^+(L_{h,i})$ for those at most $K$ buckets in level $h$ containing some member of $W_q^*$:
there are at most $\frac{bm }{\varepsilon}$ entries greater or equal to $\frac{\varepsilon}{bm } $.
For other important weight classes, the cardinality of $W_q^*$ is bounded by $bm$ and the number of important weight classes $W_q^+$ with $h(q)=h$ is bounded by $\log_2(b\varepsilon^{-1})$, as it must hold that $|W_q^+| \in [\beta m b^h, \beta m b^{h+1}]$, $|W_q^+|2^{-q} \leq \frac{\varepsilon}{q_{\max}}$ and $ |W_q^+|2^{1-q} \geq \Vert W_q^+\Vert_1 \geq \frac{\varepsilon}{q_{\max}}$ and thus 
\[ \beta m b^h \frac{q_{\max}}{\varepsilon} \leq 2^q \leq 2\beta m b^{h+1} \frac{q_{\max}}{\varepsilon} ~.\]
Thus if the estimator for $G^+( z )$ uses only the largest buckets with the largest sums, the proven bounds on contraction continue to hold, and in particular
$G^+_c( Sz ) \geq (1- \varepsilon')G^+( z )$.
To show that the dilation of $G^+_c( Sz )$ is constant we need the following lemma, which shows that the probability that an important entry of $W_q$ gets placed at a much higher level than $h(q)$ is low.
This way we can bound the contribution that entries have at higher levels.

\begin{lem}\label{LemE'}
Let $q'=\log_2(n h_{max})$. With failure probability at most $\delta $ the event $\mathcal{E}'$ holds that there is no entry $z_p \in W_q$ with $q \leq q'$ and $h_p > h_q:= h(q)+ \log_b\left(\frac{q'bm }{\delta}\right)$.
\end{lem}
{
\begin{proof}
Let $q \leq q'$.
For any coordinate $p$ and any level $h$, the probability that $h_p > h$ can be bounded by $\sum_{h'=h+1}^{h_{\max}}\frac{1}{\beta b^{h'}}\leq \frac{1}{b^h}$ since $b>2$.
Thus the expected number of coordinates $z_p \in W_q$ with $h_p > h_q$ is bounded by 
\[\frac{|W_q|}{b^{h_q}}\leq \frac{b^{h_q+1}m}{b^{h_q}}\leq \frac{\delta}{q'} .\]
This also gives us an upper bound for the probability that there is no coordinate in $W_q$ with $h_p > h_q$.
Using the union bound for all $q \leq q'$ completes the proof.
\end{proof}
}

{
\begin{proof}[Proof of Theorem \ref{Thm3.2}]
Assume that $\mathcal{E}'$ holds.
Note that the expected contribution of any entry $z_p$ is at most $ z_p h_{\max}$, and thus the expected contribution of all entries less than or equal to $\frac{1}{h_{\max} n}$ is at most $1$.
It remains to show that for each $q \leq q'$ the expected contribution of $W_q^+$ is bounded by $C \Vert W_q^+ \Vert_1$.
We consider the expected contribution of $W_q^+$ at level $h$ and distinguish three cases:

\medskip
  \noindent
   \textbf{Case 1:} $h=h(q)-k$ for $k > \log_b(N \frac{\ln(h_{\max}N) }{m})$. \\
   For $z_i \in W_q$ consider again the random variable  $X_i=1$ if $h_i=h$, and $X_i=0$ otherwise, and set $X=\sum X_i$.
The expected number of entries from $W_q^+$ at level $h$ is $\mathbb{E} (X)=\frac{|W_q^+|}{\beta b^h}\geq N\ln(h_{\max}N)$.
For the variance we have $\sum \sigma_i \leq \mathbb{E}(X) $ as $X$ is a sum of Bernoulli random variables.
Using Bernstein's inequality we get that $ | G(L_h) \cap W_q^{+} |\leq 2\beta^{-1}b^{-h}|W_q^{+}|$ with failure probability at most
\[P_1=\exp\left(\frac{-( \mathbb{E} (X))^2}{2\mathbb{E} (X)+\mathbb{E}(X)}\right)\leq \exp(- N/3).\]
   Hence we can assume that the number of entries in each bucket at level $h$ is at most  $\frac{2|W_q^{+}|}{\beta b^h N}\leq 2 \frac{b^k m}{N}=\ln(h_{\max}N) $.
   
   Next for each bucket and each $z_i \in W_q$ we consider the random variable $Z=\sum Z_i$ where $Z_i=1$ if entry $i$ is in the corresponding bucket, and $Z_i=0$ otherwise.
   Then the variance is bounded by $\frac{1}{N}\cdot 1^2 + 1 \cdot \frac{1}{N^2}\leq \frac{2}{N}$.
   Applying Bernstein's inequality gives us
   \[ P(Z\geq 2 \frac{b^k m}{N} + \lambda)\leq \exp\left(\frac{-\lambda^2}{2N \cdot \frac{2}{N}+2\lambda/3}\right). \]
   For $\lambda = \ln(h_{\max}N)$ this implies $P(Z\geq 2 \frac{b^k m}{N} + \lambda)={O}((h_{\max}N)^{-1})$.
   Using the union bound, the probability that there exists a bucket with at least $ \frac{b^k m}{N} + \lambda$ coordinates can be bounded by $P_2={O}((h_{\max})^{-1})$
   Further we have 
   \[\Vert W_q^{+}\Vert_1 \geq | W_q^{+} |2^{-q} \geq  2^{-q} b^{h(q)}\beta m.\]
   The expected contribution of $W_q^+$ at level $h$ can thus be bounded by
   \begin{align*}
   \Lambda &=
   P_2 \cdot2^{1-q}3|W_q^{+}| b^h\beta + K\left(\frac{3|W_q^{+}|}{\beta b^h N}\right)2^{1-q}b^h\beta\\
  & \leq \left({O}(h_{\max}^{-1})+ \frac{3K}{N}2\right)\Vert W_q^{+}\Vert_1={O}(h_{\max}^{-1})\Vert W_q^{+}\Vert_1.
   \end{align*}
Summing over at most $h_{\max}$ levels we have that the contribution in this case is bounded by $O(1)\Vert W_q^{+}\Vert_1$.   
   
\medskip
  \noindent
   \textbf{Case 2:} $h=h(q)+k$ for $k \geq \log_b\left(\frac{q'bm }{\delta}\right)$. \\
   By $\mathcal{E'}$ (see Lemma \ref{LemE'}), the set $L_h \cap W_q^+$ is empty, and thus the expected contribution in this case is $0$.
   
   \medskip
  \noindent
   \textbf{Case 3:} $h\geq h(q)-\log_b(\frac{N\ln(h_{\max}N)}{m})$ and $h< h(q)+\log_b\left(\frac{q'bm }{\delta}\right)$. \\
   Note that $\log_b(\frac{N\ln(h_{\max}N)}{m})+ \log_b\left(\frac{q'bm }{\delta}\right)$ is constant since $b \geq \max\left\lbrace m,~ \delta^{-1} \right\rbrace$ by Assumption \ref{Ass3.1}, and thus $N=b^{\mathcal{O}(1)}$, and the expected contribution of each level is at most constant.
   The total expected contribution is thus ${O}(1)\Vert W_q^{+}\Vert_1$.
\end{proof}
}

\begin{proof}[Proof of Theorem \ref{ThmG+}]
The result follows by combining Theorem \ref{netthm} and Theorem \ref{Thm3.2} and substituting $240\varepsilon'$ by $\varepsilon$. The $\poly(\mu d \log n)$ bound on the sketch size follows from $r=Nh_{\max}=Nh_{\max}=O(N\log n)$ and by using Assumption \ref{Ass3.1} for bounding $N$.
\end{proof}

\section{Omitted details from Section \ref{sectUniform}}\label{app:uniform}
To show Theorem \ref{Thm2.1} we will first define the notions of sensitivities, VC-dimension, and the range space induced by a set of functions.
\begin{mydef}{\cite{LangbergS10}}
	\label{def:sensitivities}
	Consider a family of functions $\mathcal{F}=\{g_1,\ldots,g_n\}$ mapping from $\mathbb{R}^d$ to $[0,\infty)$ and weighted by $w\in\mathbb{R}_{> 0}^n$. The sensitivity of $g_i$ for $f_w(x)=\sum\nolimits_{i=1}^{n} w_i g_i(x)$ is
	\begin{align}
	\label{eqn:sensitivities}
	\varsigma_i = \sup \frac{w_i g_i(x)}{f_w(x)}
	\end{align}
	where $\sup$ is over all $x\in\mathbb{R}^d$ with $f_w(x) > 0$. If this set is empty then $\varsigma_i=0$. The total sensitivity is $\mathfrak{S} = \sum\nolimits_{i=1}^{n} \varsigma_i$.
\end{mydef}

The sensitivity of a point bounds the maximal relative contribution to the target function the point can have.
Computing the sensitivities is often intractable and necessitates approximating the original optimization problem close to optimality. However, this is the problem that we want to solve, see \cite{BravermanFL16}.
Fortunately, for our applications it suffices to obtain a reasonable upper bound for the sensitivities.

\begin{mydef}
	A range space is a pair $\mathfrak{R}=(\mathcal{F},\ranges)$ where $\mathcal{F}$ is a set and $\ranges$ is a family of subsets of $\mathcal{F}$. The VC dimension $\Delta(\mathfrak{R})$ of $\mathfrak{R}$ is the size $|G|$ of the largest subset $G\subseteq \mathcal{F}$ such that $G$ is shattered by $\ranges$, i.e., $\left| \{G\cap R\mid R\in \ranges \} \right| = 2^{|G|}.$
\end{mydef}
\begin{mydef}
	Let $\mathcal{F}$ be a finite set of functions mapping from $\mathbb{R}^d$ to $\mathbb{R}_{\geq 0}$. For every $x\in\mathbb{R}^d$ and $r\in \mathbb{R}_{\geq 0}$, let $\rng{\mathcal{F}}(x,r) = \{ f\in \mathcal{F}\mid f(x)\geq r\}$, and $\ranges(\mathcal{F})=\{\rng{\mathcal{F}}(x,r)\mid x\in\mathbb{R}^d, r\in \mathbb{R}_{\geq 0} \}$, and $\mathfrak{R}_{\mathcal{F}}=(\mathcal{F},\ranges(\mathcal{F}))$ be the range space induced by $\mathcal{F}$.
\end{mydef}

The VC-dimension can be thought of as something similar to the dimension of our problem.
For example the VC-dimension of the set of hyperplane classifiers in $\mathbb{R}^d$ is $d+1$.
The sensitivity scores were combined with a theory on the VC dimension of range spaces in \cite{BravermanFL16}. We employ a more recent version from \cite{FeldmanSS20}.

\begin{pro}{\cite{FeldmanSS20}}
	\label{thm:sensitivity}
	Consider a family of functions $\mathcal{F}=\{f_1,\ldots,f_n\}$ mapping from $\mathbb{R}^d$ to $[0,\infty)$ and a vector of weights $w\in\mathbb{R}_{> 0}^n$. Let $\varepsilon,\delta\in(0,1/2)$.
	Let $s_i\geq \varsigma_i$.
	Let $S=\sum\nolimits_{i=1}^{n} s_i \geq \mathfrak{S}$. Given $s_i$ one can compute in time $O(|\mathcal{F}|)$ a set $R\subset \mathcal{F}$ of $$O\left( \frac{S}{\varepsilon^2}\left( \Delta \ln S + \ln \left(\frac{1}{\delta}\right) \right) \right)$$ weighted functions such that with probability $1-\delta$, we have for all $x\in \mathbb{R}^d$ simultaneously $$\left| \sum_{f_i\in \mathcal{F}} w_i f_i(x) - \sum_{f_i\in R} u_i f_i(x) \right| \leq \varepsilon \sum_{f_i\in \mathcal{F}} w_i f_i(x),$$
	where each element of $R$ is sampled i.i.d. with probability $p_j=\frac{s_j}{S}$ from $\mathcal{F}$, $u_i = \frac{Sw_j}{s_j|R|}$ denotes the weight of a function $f_i\in R$ that corresponds to $f_j\in\mathcal{F}$, and where $\Delta$ is an upper bound on the VC dimension of the range space $\mathfrak{R}_{\mathcal{F}^*}$ induced by $\mathcal{F}^*$ obtained by defining $\mathcal{F}^*$ to be the set of functions $f_j\in\mathcal{F}$, where each function is scaled by $\frac{Sw_j}{s_j|R|}$.
\end{pro}


Now we show how Proposition \ref{thm:sensitivity} can be used to approximate our loss function on the negative domain. We define $g_i(x)=\min\lbrace g(a_i x) ,\ln(2)\rbrace$.

\begin{lem}\label{lem:rngsp}
The range space induced by $\mathcal{F}=\{ g_i ~|~ i \in [n] \}$ satisfies $\Delta(\mathcal{R}_{\mathcal{F}})\leq d+1$.
\end{lem}
{
\begin{proof}
Note that $g $ is invertible and monotone.
Let $G \subseteq \mathcal{F},$ $x \in \mathbb{R}^d$ and $r \in \mathbb{R}$.
For $r> \ln(2)$ we have  $\range_G(x, r)= \emptyset$, otherwise (for $r \leq \ln(2)$) we have 
\begin{align*}
\range_G( x, r)&=\{g_i \in G ~|~ g_i(x)\geq r \}\\
&= \{g_i \in G ~|~ a_i x \geq g^{-1}(r) \}.
\end{align*}
Hence
\begin{align*}
~|\{ \range_G(x, r) ~|~ x \in \mathbb{R}^d, r\in \mathbb{R}_{\geq 0} \} |
&= \left\vert\lbrace \{ g_i \in G ~|~ a_i x -g^{-1}(r) \geq 0  \} ~|~ x \in \mathbb{R}^d, r\leq \ln(2) \rbrace \lbrace \emptyset \rbrace \right\vert \\
& \leq  \left\vert\{ \lbrace g_i \in G ~|~ a_i x -\tau \geq 0  \} ~|~ x \in \mathbb{R}^d, \tau \in \mathbb{R} \rbrace \right\vert.
\end{align*}
The last set is the set of points that is shattered by the affine hyperplane classifier $a_i \mapsto \mathbf{1}_{a_ix - \tau\geq 0}$.
Its VC dimension is thus $d+1$ \cite{KearnsV94}, implying $ |\{ \range_G(x, r) ~|~ x \in \mathbb{R}^d, r\in \mathbb{R}_{\geq 0} \} | =2^{|G|}$ can only hold if $ |G|\leq d + 1$,  and thus the VC dimension of $\mathcal{F}$ is at most $d+1$.
\end{proof}
}

We set $f_{\min}(x)= \frac{n}{\mu}$ for all $x \in \mathbb{R}^d$.

\begin{cor}\label{cor2}
The range space induced by $\mathcal{F}=\{ g_i ~|~ i \in [n] \} \cup \{ f_{\min} \}$ satisfies $\Delta(\mathcal{R}_{\mathcal{F}})\leq d+2$.
\end{cor}

{
\begin{proof}
Assume there exists $G \subset \mathcal{F}$ with $|G|\geq d+3$ and $ \{G\cap R\mid R\in \ranges(\mathcal{F}) \}  = 2^{G}.$
Then we have for $G'=G\setminus \{f_{\min}\}$ $|G'|\geq d+2$ and $ \{G'\cap R\mid R\in \ranges(\mathcal{F}) \}  = 2^{G'}$ contradicting Lemma \ref{lem:rngsp}.
\end{proof}
}

Now we are ready to prove Theorem \ref{Thm2.1}.

\begin{proof}[Proof of Theorem \ref{Thm2.1}]
We want to apply Proposition \ref{thm:sensitivity} to $\mathcal{F}=\{ g_i ~|~ i \in [n] \} \cup \{ f_{\min} \}$.
By Corollary \ref{cor2} the VC-dimension of $\mathcal{F}$ is at most $d+2$.
Note that the sensitivity of any function, in particular of $f_{\min}$, is at most $1$ 
and for the sensitivity of any function other than $f_{\min}$ we have
\[\sup_{x\in \mathbb{R}^d\setminus\{0\}}\frac{g_i(x)}{f((Ax)^-)+f_{\min}(x)}\leq \sup_{x}\frac{g_i(x)}{f_{\min}(x)}\leq \frac{\ln(2)\mu}{n}\]
and the total sensitivity is thus bounded by $\ln(2)\mu + 1$.
Hence we can use Proposition \ref{thm:sensitivity} and with failure probability at most $\delta_1$, we get a subset $r$ of $\mathcal{F}=\{ g_i ~|~ i \in [n] \} \cup \{ f_{\min} \}$ and a weight vector $u$ such that
\[ \left| \sum_{f_i\in \mathcal{F}} f_i(x) - \sum_{f_i\in R} u_i f_i(x) \right| \leq \varepsilon \sum_{f_i\in \mathcal{F}} f_i(x).\]
As it holds that $\sum_{f\in \mathcal{F}} f_i(x)=f((A x)^-)+ \frac{n}{\mu}$, and since the weight of $f_{\min}$ is $1$, this implies for $R'=R \setminus \{f_{\min}\}$ that we have an error $\Delta$ of at most
\begin{align*}
\Delta&=\left| \sum_{f_i\in \mathcal{F}\setminus \{f_{\min}\}} f_i(x) - \sum_{f_i\in R'} u_i f_i(x) \right| 
\leq \varepsilon \sum_{f_i\in \mathcal{F}} f_i(x)
\leq \varepsilon \left(\sum_{f_i\in \mathcal{F}\setminus \{f_{\min}\}} f_i(x) ~+~ \frac{n}{\mu} \right)
\leq 3\varepsilon f(A x).
\end{align*}
The last inequality follows from Lemma \ref{minvalue}, and using that $ f(A x) \geq \sum_{f_i\in \mathcal{F}\setminus \{f_{\min}\}} f_i(x)$.
This proves the first part of the theorem.

Observe that the expected contribution of row $a_i$ is $P(g_i \in R)a_i u_i=\frac{k}{n} \frac{n}{k}g(a_i x)=g(a_i x)$.
Thus the second statement follows by linearity of expectation. 
\end{proof}

\section{Omitted details from Section \ref{sec:logreg}}
\begin{proof}[Proof of Corollary \ref{mainthm}]
We need one streaming pass over the data in ${O}(\nnz(A))$ time to draw a uniform random sample $T$ from Theorem \ref{Thm2.1} and to compute $A'=SA$.
Now compute the $x^*$ minimizing the convex objective function
$f(Tx^*)+G^+_c(A'x^*)$.
This can be done using the ellipsoid method on the following convex program:
we have one variable $x_i$ for $i \in [d]$.
For each row $t_i$ of $T$, construct a variable $v_i$ and a constraint $v_i \geq t_i x$.
Similarly for each row $a_i'$ of $A'$ construct a variable $v_i'\geq 0$ and a constraint $v_i' \geq a_i' x$. The objective is to minimize the convex function
\[ \sum g(v_i) + \sum v_i' .\]
The convex program has $\poly(\mu d \log n)$ many variables, thus the running time is also $\poly(\mu d \log n)$.
Using the same analysis as in the previous proof shows that the solution $x'$ we get satisfies 
\[f(Ax')\leq {O}(1) \min_{x\in \mathbb{R}^d} f(Ax)\]
with constant probability.
\end{proof}

\section{Omitted details from Section \ref{sec:experiments}}\label{sec:supp:experiments}

\begin{table*}[!htbp]
	\centering
	\begin{tabular}{ | l| l| l| l| l| }
		\hline
		{\bf Dataset} & {$n$} & $d$ & {\bf Source} \\ \hline \hline
		Covertype & $581012$ & $54$ & \url{https://archive.ics.uci.edu/ml/datasets/Covertype} \\ \hline 
		Webspam & $350000$ & $128$ & \url{https://www.csie.ntu.edu.tw/~cjlin/libsvmtools/datasets/binary.html#webspam} \\ \hline 
		Kddcup & $494021$ & $33$ & \url{https://kdd.ics.uci.edu/databases/kddcup99/kddcup99.html} \\ \hline 
		Synthetic & $100000$ & $2$ & - \\ \hline 
	\end{tabular}
	\caption{Summary of the datasets: $d$ is given without intercept. Datasets are downloaded resp. generated automatically by our open code.}
	 \label{tab:datasetvalues}
\end{table*}

\begin{figure*}[ht!]
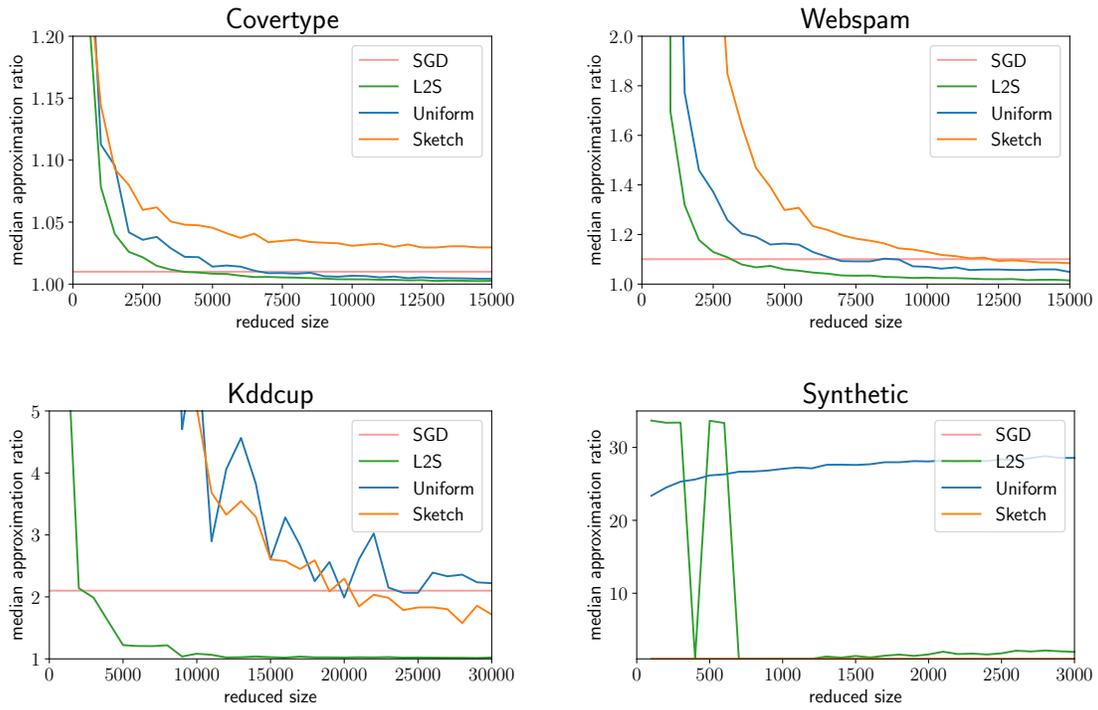

\begin{center}
\begin{tabular}{cc}
\includegraphics[width=0.44\linewidth]{newfigures/covertype_sklearn_ratio_plot.pdf}&
\includegraphics[width=0.44\linewidth]{newfigures/webspam_libsvm_desparsed_ratio_plot.pdf}\\
\includegraphics[width=0.44\linewidth]{newfigures/kddcup_sklearn_ratio_plot.pdf}&
\includegraphics[width=0.44\linewidth]{newfigures/synthetic_n_100000_ratio_plot.pdf}
\end{tabular}
\caption{Comparison of the approximation ratios.
}
\end{center}
\end{figure*}

\begin{figure*}[ht!]
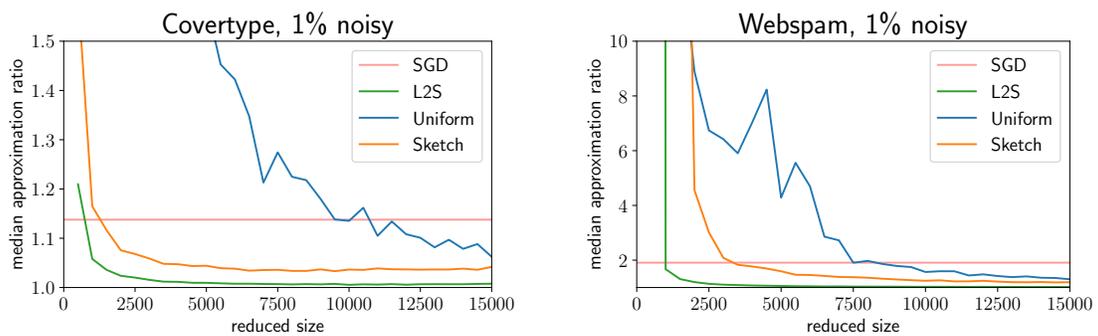

\begin{center}
\begin{tabular}{cc}
\includegraphics[width=0.44\linewidth]{newfigures/covertype_sklearn_noisy_ratio_plot.pdf}&
\includegraphics[width=0.44\linewidth]{newfigures/webspam_libsvm_desparsed_noisy_ratio_plot.pdf}
\end{tabular}
\caption{Comparison of the approximation ratios with added noise.
}
\end{center}
\end{figure*}

\begin{figure*}[ht!]
\begin{center}
\begin{tabular}{cc}
\includegraphics[width=0.44\linewidth]{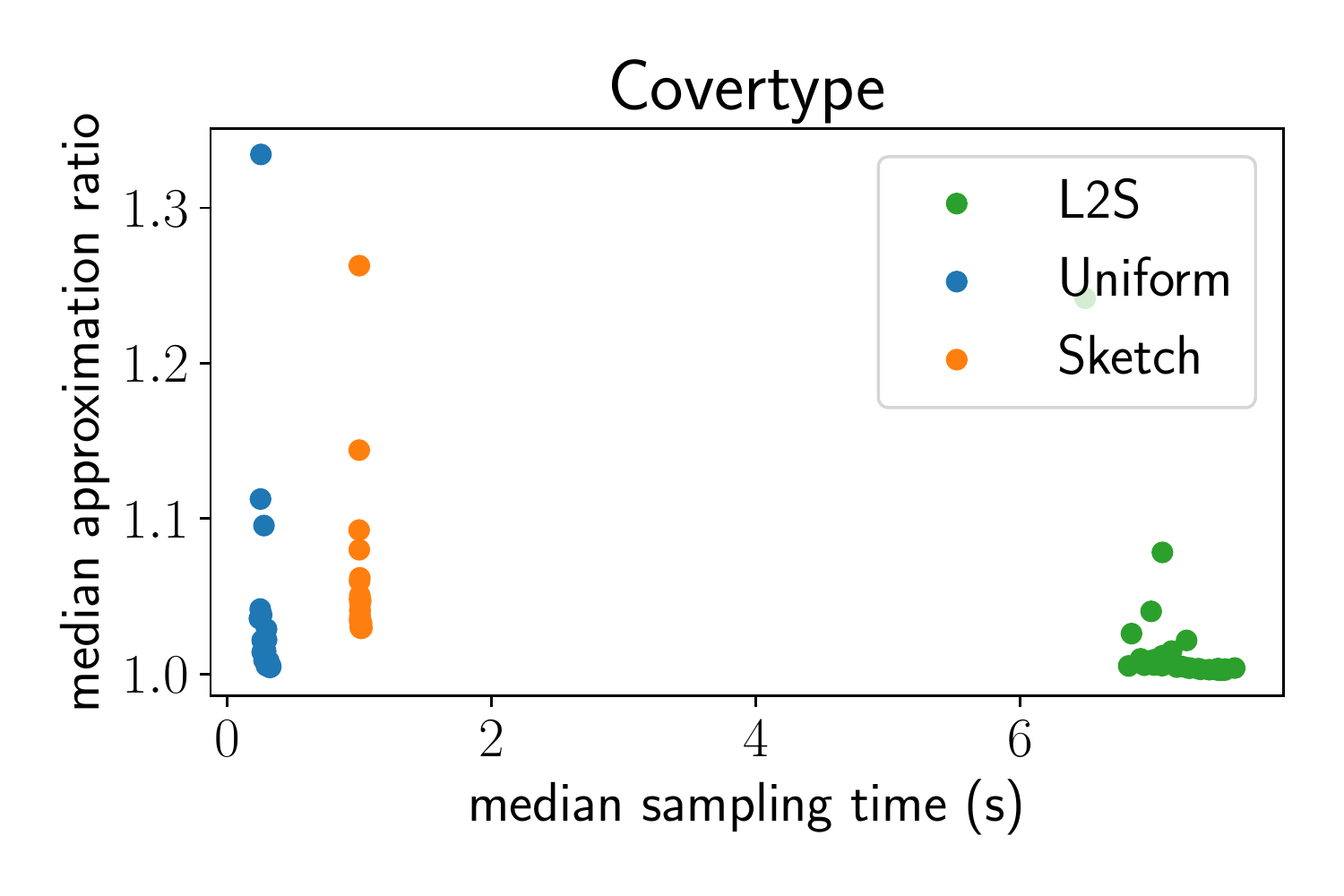}&
\includegraphics[width=0.44\linewidth]{newfigures/webspam_libsvm_desparsed_sampling_time_plot.pdf}\\
\includegraphics[width=0.44\linewidth]{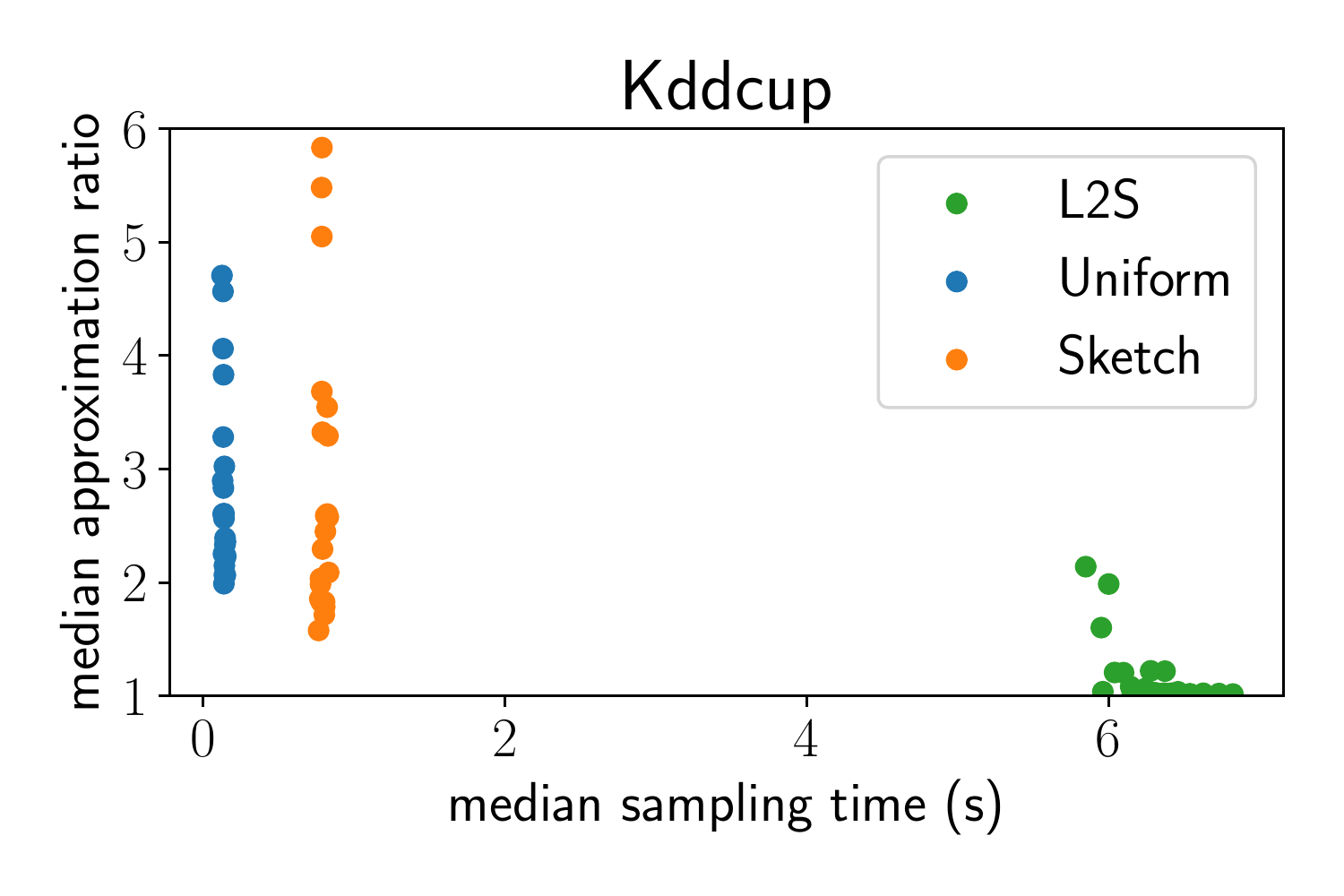}&
\includegraphics[width=0.44\linewidth]{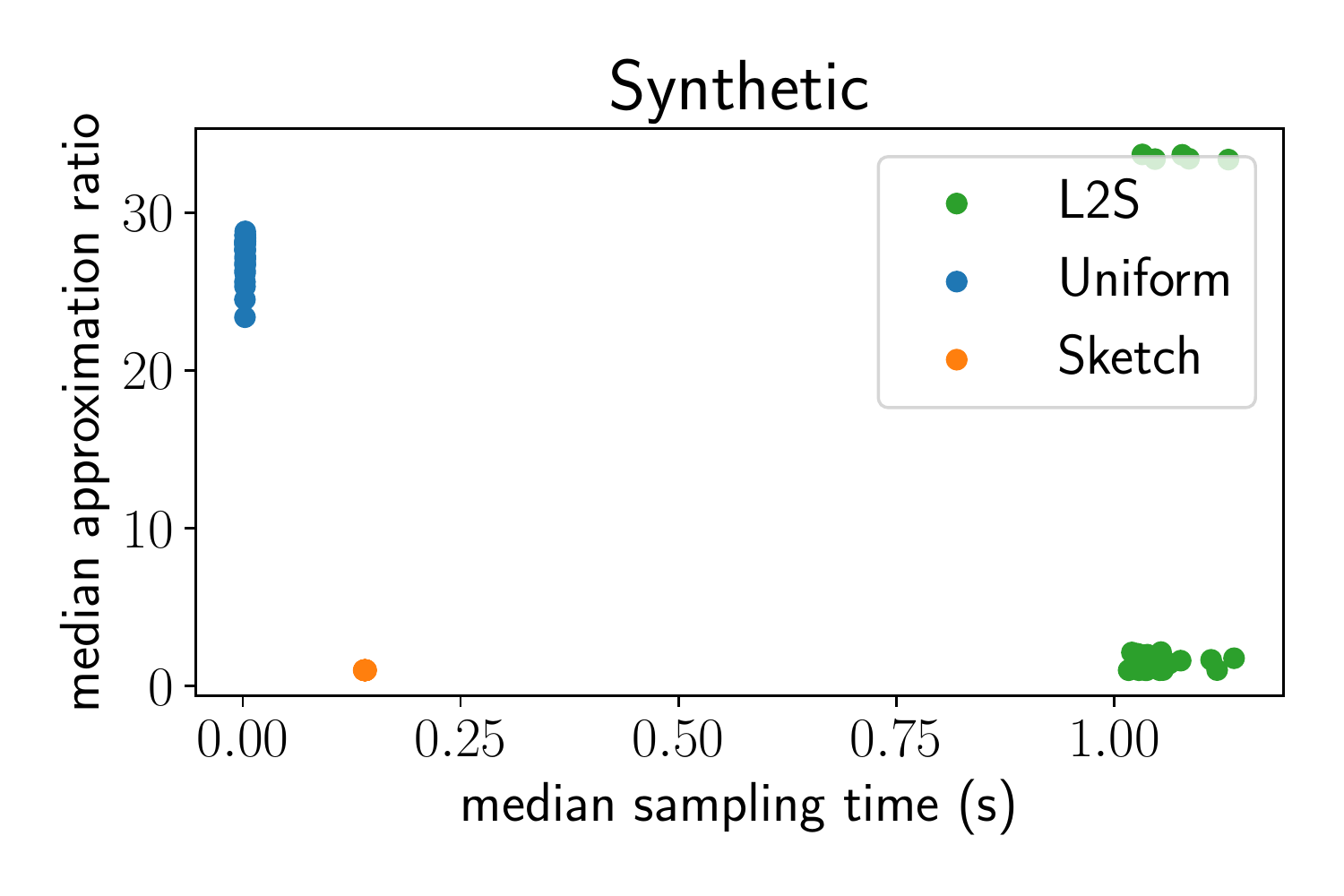}
\end{tabular}
\caption{Comparison of sketching resp. sampling times vs. accuracy.
}
\end{center}
\end{figure*}

\begin{figure*}[ht!]
\begin{center}
\begin{tabular}{cc}
\includegraphics[width=0.44\linewidth]{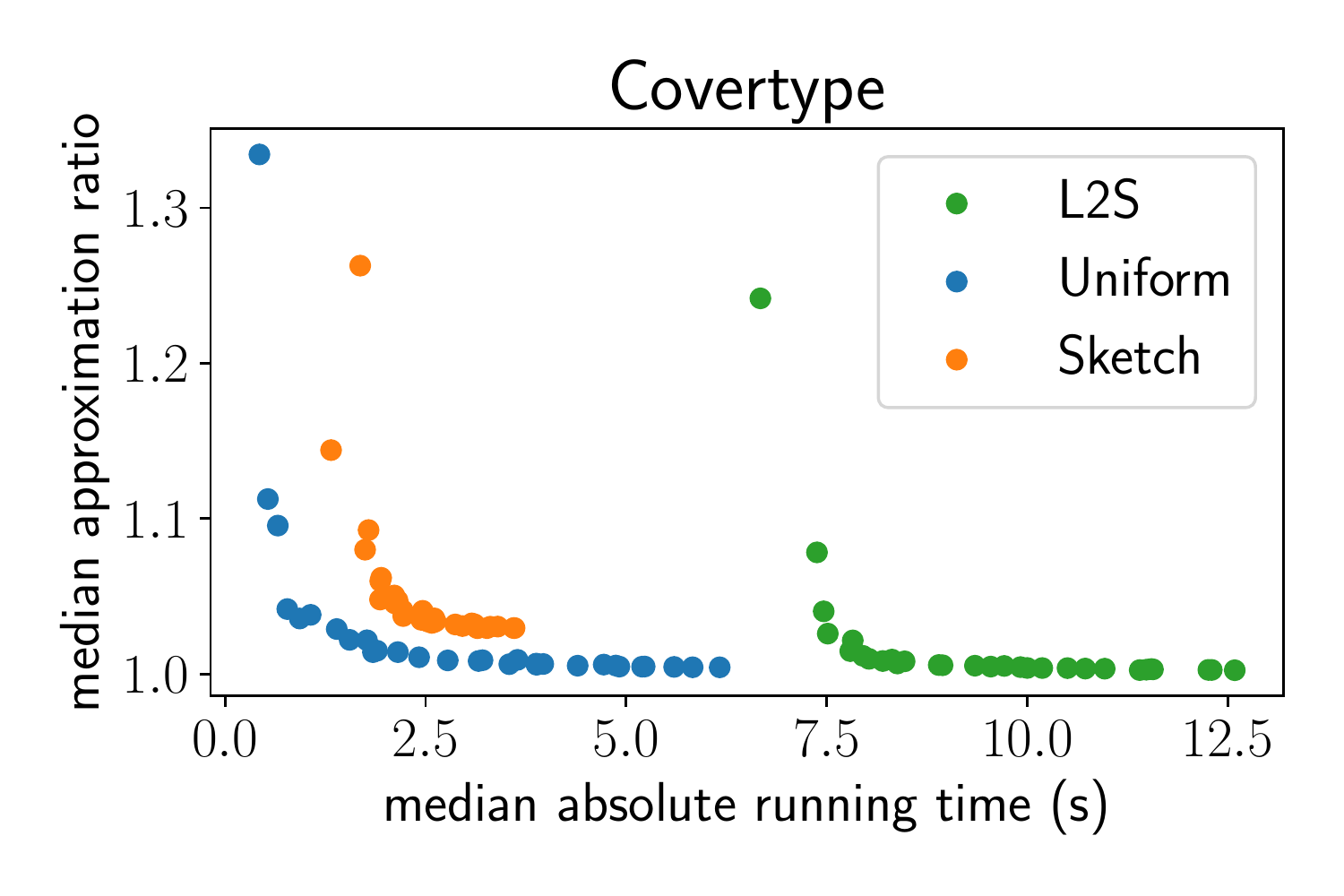}&
\includegraphics[width=0.44\linewidth]{newfigures/webspam_libsvm_desparsed_total_time_plot.pdf}\\
\includegraphics[width=0.44\linewidth]{newfigures/kddcup_sklearn_total_time_plot.pdf}&
\includegraphics[width=0.44\linewidth]{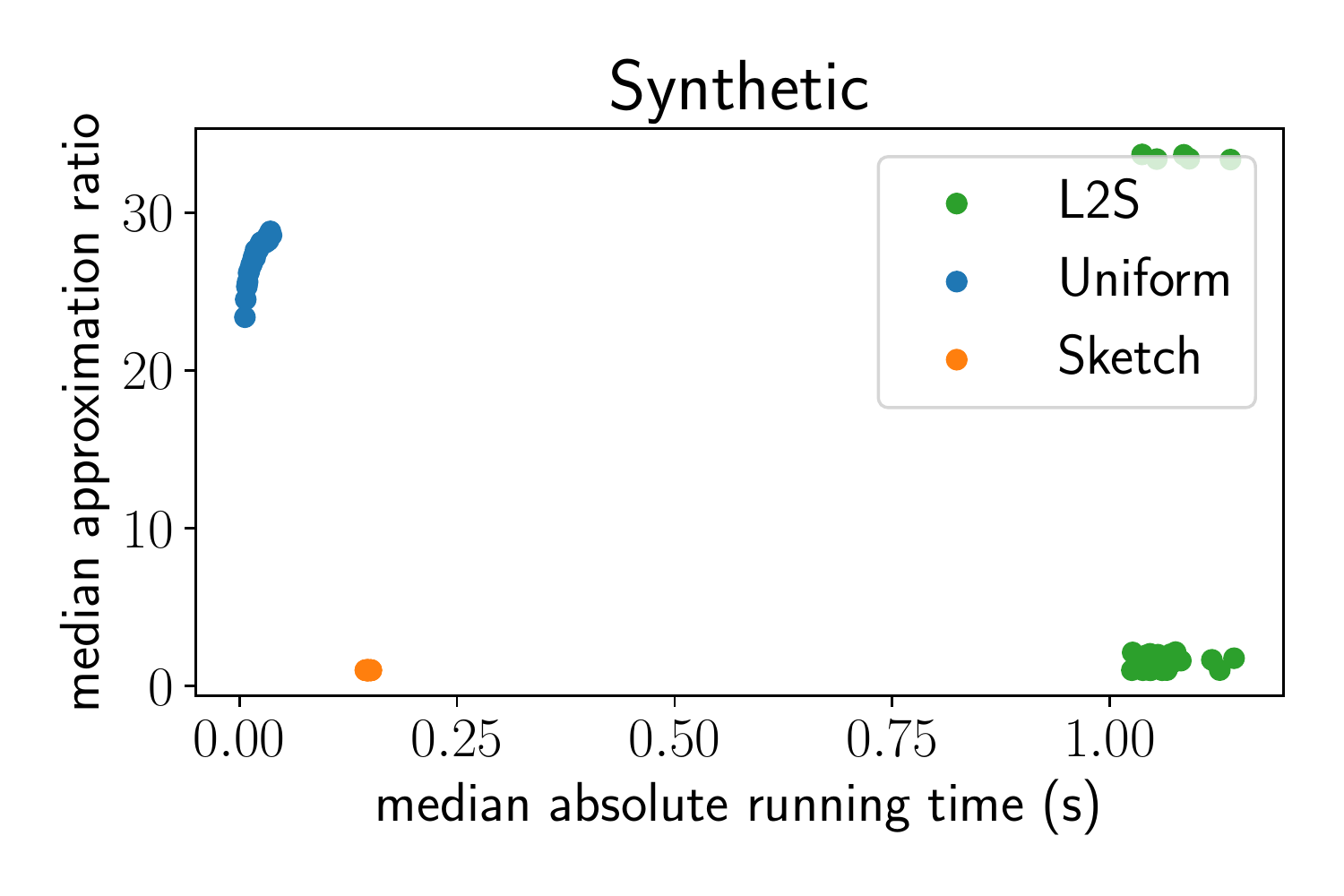}
\end{tabular}
\caption{Comparison of total running times including optimization vs. accuracy.
}
\end{center}
\end{figure*}

\end{document}